\documentclass[a4paper,UKenglish,cleveref, autoref, thm-restate]{lipics-v2021}

\hideLIPIcs  

\usepackage{graphicx} 
\usepackage[T1]{fontenc}
\usepackage{amsmath}
\usepackage{amssymb}
\usepackage{amsthm}

\usepackage{thmtools}
\usepackage{mathtools}
\usepackage{graphics}

\usepackage{todonotes}
\usepackage{hyperref}
\usepackage{tikzit}

\tikzstyle{new style 0}=[fill=black, draw=black, shape=circle, scale=0.4]
\tikzstyle{new style 1}=[fill=red, draw=black, shape=rectangle]
\tikzstyle{new style 2}=[fill={rgb,255: red,191; green,255; blue,0}, draw={rgb,255: red,191; green,255; blue,0}, shape=circle]

\tikzstyle{new edge style 0}=[-, draw=red, line width=4pt]
\tikzstyle{new edge style 1}=[-, draw={rgb,255: red,191; green,255; blue,0}, line width=4pt]
\tikzstyle{new edge style 2}=[-, fill={rgb,255: red,255; green,128; blue,0}]
\tikzstyle{new edge style 3}=[-, draw={rgb,255: red,255; green,128; blue,0}, line width=3pt]
\tikzstyle{new edge style 4}=[-, fill=green, draw=green]
\tikzstyle{new edge style 5}=[-, fill=red, draw=red]
\tikzstyle{new edge style 6}=[-, fill=blue, draw=blue]
\tikzstyle{new edge style 7}=[-, dashed]
\tikzstyle{new edge style 8}=[-, fill=none, line width=0.5pt]
\tikzstyle{new edge style 9}=[-, line width=2pt, draw=red]

\usepackage{algorithm}
\usepackage[noend]{algpseudocode}

 \declaretheorem[name={Open Problem}]{problem}

\newcommand\floor[1]{\left\lfloor#1\right\rfloor}
\newcommand\ceil[1]{\left\lceil#1\right\rceil}

\newcommand{\poly}{\mathrm{poly}}
\newcommand{\polylog}{\mathrm{polylog}}
\newcommand{\Oh}[1]{\mathcal{O}{\left(#1\right)}}

\newcommand{\Ot}[1]{\widetilde{\mathcal{O}}{\left(#1\right)}}

\newcommand{\eps}{\varepsilon}

\newcommand{\sdiff}{\bigtriangleup}
\newcommand{\Rr}{\mathcal{R}}
\newcommand{\dist}{\mathrm{dist}}
\newcommand{\diam}{\mathrm{diam}}

\newcommand{\AddNeighbours}{\textnormal{\textsc{AddNeighbours}}}
\newcommand{\ListDifferences}{\textnormal{\textsc{ListDifferences}}}

\newcommand{\ExpandBalls}{\textnormal{\textsc{ExpandBalls}}}
\newcommand{\ExpandBallsRecursively}{\textnormal{\textsc{ExpandBallsRecursively}}}
\newcommand{\InitTree}{\textnormal{\textsc{InitTree}}}

\newcommand{\UpdateBot}{\textnormal{\textsc{UpdateBot}}}
\newcommand{\UpdateTop}{\textnormal{\textsc{UpdateTop}}}
\newcommand{\Push}{\textnormal{\textsc{Push}}}
\newcommand{\PushTop}{\textnormal{\textsc{PushTop}}}
\newcommand{\PushBot}{\textnormal{\textsc{PushBot}}}
\newcommand{\Initialize}{\textnormal{\textsc{Initialize}}}
\newcommand{\Add}{\textnormal{\textsc{Add}}}

\newcommand{\Mark}{\textnormal{\textsc{Mark}}}

\title{Better Diameter Algorithms for Bounded VC-dimension Graphs and Geometric Intersection Graphs} 

\titlerunning{Better Diameter Algorithms for Bounded VC-dimension and Intersection Graphs} 

\author{Lech Duraj}{Faculty of Mathematics and Computer Science, Jagiellonian University in Kraków, Poland}{lech.duraj@uj.edu.pl}{https://orcid.org/0000-0002-0004-3751}{}

\author{Filip Konieczny}{Faculty of Mathematics and Computer Science, Jagiellonian University in Kraków, Poland}{filip.konieczny@student.uj.edu.pl}{https://orcid.org/0009-0001-9397-0289}{During the preparation of this article, Filip Konieczny was a participant of the tutoring programme under the Excellence Initiative at the Jagiellonian University.}

\author{Krzysztof Potępa}{Faculty of Mathematics and Computer Science, Jagiellonian University in Kraków, Poland}{krzysztof.potepa@doctoral.uj.edu.pl}{https://orcid.org/0009-0002-8230-6961}{}

\authorrunning{L. Duraj, F. Konieczny, K. Potępa} 

\Copyright{Lech Duraj, Filip Konieczny, Krzysztof Potępa} 



\funding{All authors were supported by National Science Center of Poland grant 2019/34/E/ST6/00443.}

\acknowledgements{We would like to thank the reviewers for helping to improve our paper with their suggestions.}

\nolinenumbers 

\EventEditors{Timothy Chan, Johannes Fischer, John Iacono, and Grzegorz Herman}
\EventNoEds{4}
\EventLongTitle{32nd Annual European Symposium on Algorithms (ESA 2024)}
\EventShortTitle{ESA 2024}
\EventAcronym{ESA}
\EventYear{2024}
\EventDate{September 2--4, 2024}
\EventLocation{Royal Holloway, London, United Kingdom}
\EventLogo{}
\SeriesVolume{308}
\ArticleNo{10}

\begin{document}

\maketitle

\begin{abstract}
We develop a framework for algorithms finding the diameter in graphs of bounded distance Vapnik-Chervonenkis dimension, in (parameterized) subquadratic time complexity. The class of bounded distance VC-dimension graphs is wide, including, e.g. all minor-free graphs.

We build on the work of Ducoffe et al. [SODA'20, SIGCOMP'22], improving their technique. With our approach the algorithms become simpler and faster, working in $\Oh{k \cdot n^{1-1/d} \cdot m \cdot \polylog(n)}$ time complexity for the graph on $n$ vertices and $m$ edges, where $k$ is the diameter and $d$ is the distance VC-dimension of the graph. Furthermore, it allows us to use the improved technique in more general setting. In particular, we use this framework for geometric intersection graphs, i.e. graphs where vertices are identical geometric objects on a plane and the adjacency is defined by intersection. Applying our approach for these graphs, we partially answer a question posed by Bringmann et al. [SoCG'22], finding an $\Oh{n^{7/4} \cdot \polylog(n)}$ parameterized diameter algorithm for unit square intersection graph of size $n$, as well as a more general algorithm for convex polygon intersection graphs.
\end{abstract}

\thispagestyle{empty}
\clearpage
\setcounter{page}{1}

\section{Introduction}
\label{sec:introduction}
The \emph{diameter} of a graph is the maximum possible distance between a pair of vertices. It~is~believed to be an important graph parameter and as such, it has been extensively studied. Formally, the \textsc{Diameter} and \textsc{$k$-Diameter} problems are  defined as follows\footnote{In this work, we assume graphs to be unweighted and undirected.}:

\begin{itemize}
    \item \textsc{Diameter}: Given a graph $G = (V,E)$, calculate $\diam(G) := \max_{u,v \in V} \dist(u,v)$, where $\dist(u,v)$ is the shortest possible length of any path between $u$ and $v$;
    \item \textsc{$k$-Diameter}: Given a graph $G = (V,E)$ and $k \in \mathbb{Z}$, determine whether $\diam(G) \leq k$.
\end{itemize} 

Both \textsc{Diameter} and \textsc{$k$-Diameter} are easy to solve in $\Oh{n m}$ time complexity, where $n = |V|$, $m = |E|$, by simply invoking BFS from every vertex. This straightforward algorithm, however, turns out to be notoriously hard to improve in terms of time complexity. In 2013, Roditty and Vassilevska-Williams showed \cite{RV13} that any algorithm solving \textsc{2-Diameter} in $\Oh{m^{2-\varepsilon}}$ time complexity would imply the existence of a $(2-\delta)^n$ time algorithm for SAT, thus refuting Strong Exponential Time Hypothesis (SETH, \cite{IP01}). Although conditional, it is an argument for the existence of a quadratic complexity barrier. Furthermore, the hardness of \textsc{2-Diameter} implies that approximating the diameter with ratio better than $3/2$ in subquadratic time would refute SETH as well.
But even if we assume SETH to be true, there is still a lot of open questions about diameter. One long line of research deals with subquadratic approximation in general graphs, and trade-offs between complexity and approximation ratio \cite{RV13,CLRTW14,CGR16,Li20,Li21, DLV21, DW21, Bo22, ADLV23}. 

Another family of questions arises from considering the diameter problem for some restricted graph classes (\cite{AVW16,CFMP01,Ol90,FP80}. A notable example is the case of planar graphs: the first subquadratic algorithm was found by Cabello \cite{Ca18}, and the fastest currently known works in $\Ot{n^{5/3}}$ time\footnote{The $\Ot{}$ notation ignores logarithmic factors, i.e. $\Ot{f(n)}$ means $\Oh{f(n) \cdot \polylog(f(n))}$.} and is due to Gawrychowski et al. \cite{GKMSW21}. Similar problems arise from considering \emph{geometric intersection graphs}: we take a family of objects in $\mathbb{R}^k$, name them the vertices of our graph, and define an edge between a pair of objects to exist if and only if they intersect. There is a natural interpretation of a diameter problem for these graphs, especially if the objects are unit balls or axis-aligned unit squares on the plane: the middle point of each object is a communication node, and the object itself represents its maximal range of communication. The diameter of the graph is the maximal number of hops needed for any two nodes to successfully communicate. Observe that the resulting graph on $n$ objects can easily have $\Theta(n^2)$ edges, as well as very large cliques. This makes a subquadratic algorithm somewhat more tricky, as we cannot ever list the edges of this graph explicitly, but instead we have to rely on its geometric representation. Geometric intersection graphs have also been studied in terms of fine-grained complexity \cite{CS16,CS17}, but there were relatively few subquadratic breakthroughs for the diameter problem. A recent paper by Bringmann et al. \cite{BKKNZ22}, proved (among other results) that:
\begin{itemize}
\item neither the intersection graph of axis-parallel unit cubes nor unit balls in $\mathbb{R}^3$ admits a subquadratic diameter algorithm under SETH;
\item the intersection graph of axis-parallel unit cubes in $\mathbb{R}^{12}$ does not admit, under the Hyperclique Hypothesis, a subquadratic algorithm for 2-\textsc{Diameter};
\item for the intersection graph of axis-parallel unit squares in $\mathbb{R}^2$ there is an algorithm with $\Oh{n \log n}$ time complexity for 2-\textsc{Diameter}.
\end{itemize}

Which other graph classes are non-trivial to consider in this setting? Some good choices are, for example, $K_t$-minor-free graphs, as they forbid using the counterexample from \cite{RV13}.  Ducoffe, Habib and Viennot \cite{DHV20, DHV22} proposed a more general class of graphs to consider: the ones with bounded \emph{distance Vapnik-Chervonenkis dimension}, or \emph{distance VC-dimension} for short. We formally define it in section \ref{sec:preliminaries}, but roughly speaking, graph has distance VC-dimension bounded by $d$, if for every subset $A \subseteq V$ with $|A| > d$ there exists $A' \subseteq A$ which cannot be expressed as a projection of a ball, i.e. in the form $A' = \{x \in A : dist(x,v) \leq k\}$ for some $v \in V$ and $k \in \mathbb{Z}$. The class of bounded distance VC-dimension graphs includes in particular minor-free graphs (with planar graphs), interval graphs, and also geometric intersection graphs.
In their work, the authors of \cite{DHV22} showed a number of important results tying diameter finding to distance VC-dimension, in particular:
\begin{itemize}
\item a subquadratic algorithm for $k$-\textsc{Diameter}, working in $\Ot{k n^{1-\varepsilon_d} m}$ time complexity, where $\varepsilon_d \sim \frac{1}{2^d \poly(d)}$ is some (small) constant\footnote{In fact, the $2^d$ factor in \cite{DHV22} is due to considering directed graphs and other technicalities; we believe that the authors could instead claim $\Oh{d^2}$, albeit with a large multiplicative constant.};
\item a subquadratic algorithm for \textsc{Diameter}, for graphs with bounded VC-dimension which additionally admit sublinear separators (e.g. minor-closed graphs).
\end{itemize}

While preparing this version of our paper, we discovered an independent work by Hsien-Chih Chang, Jie Gao and Hung Le \cite{CGL24}. Their main result is a subquadratic algorithm which computes an additive-constant approximate (+2) diameter of a geometric intersection graph for any family of \emph{pseudo-disks} (the pseudo-disks are shapes bounded by a Jordan curve with a property that two such boundaries can have at most two intersection points; in particular, the graphs considered in this paper fit into that category). In Section \ref{sec:conclusion} we discuss how our contributions are related.

\subsection{Our contribution and paper structure}

An inspiration for this paper was to answer the open questions posed in \cite{BKKNZ22}; especially, to find a (parameterized) subquadratic algorithm for some geometric intersection graphs. In the most appealing cases of planar unit disk and unit square intersection graphs, it is not hard to prove that both these classes have their distance VC-dimension bounded by 4. Therefore, algorithms from \cite{DHV22} could in theory be applied to them, but it is impossible to do it directly, as those algorithms work only for explicitly-given sparse graphs. Therefore, we need to refine this algorithm to work in our setting.

The core idea of \cite{DHV22} is to find a spanning path of a graph $G = (V,E)$ with a low \emph{stabbing number}, i.e. an order $v_1, \ldots, v_n$ on the vertices of the graph such that for every $v \in V$, $k \in \mathbb{Z}$, every ball $N^k[v]$ can be expressed as a sum of $\Oh{n^{1-\varepsilon}}$ intervals $(v_x, v_{x+1},\ldots, v_y)$. The existence of such a path is in turn based on the results of Chazelle and Welzl \cite{CW89}, who provide a Monte Carlo polynomial algorithm for finding such a path. The authors of \cite{DHV22} use this algorithm as a subroutine (``black-box''), employing a neat trick to bring down its polynomial complexity to a subquadratic one.

Interestingly, the Chazelle-Welzl subroutine uses a technique similar to the main construction of \cite{DHV22} -- in particular, the notion of $\varepsilon$-nets, first introduced in \cite{HW86}. In this paper we show that these two constructions can be, in a natural way, replaced by only one argument. This requires going back on the basic definitions, in particular relaxing the conditions on the stabbing number, as well as different complexity analysis. We are, however, rewarded with a simpler and more straightforward algorithm, naturally working in subquadratic time. Furthermore, this also brings down the time complexity of the algorithm, and opens new possibilities of its generalization. 

The high-level concept is as follows: for any $j$ and for a vertex $v$ let us denote by $N^j[v]$ the $j$-neighbourhood of $v$, i.e. all vertices reachable from $v$ via at most $j$ edges. We find a particular order $v_1, \ldots, v_n$ on all the vertices  such that  for every $i$ the neighbourhoods $N^j[v_i]$ and $N^j[v_{i+1}]$ differ relatively little -- to be precise, the total size  of the difference sets $N^j[v_i] \sdiff N^j[v_{i+1}]$ is subquadratic. This order can be found for every $j$ with a randomized algorithm, using the concept of \emph{$\varepsilon$-nets}, and it allows us to encode all the $j$-neighbourhoods in subquadratic space. It is now enough to devise a way to compute this encoding also in subquadratic time; similarly to \cite{DHV22}, we do it incrementally, going from all $N_{j-1}[v]$ sets to all $N_{j}[v]$ sets. We need, however, different constructions for general sparse graphs (when trying to improve \cite{DHV22}) and for implicitly given graphs (like geometric intersections). For the latter, we show that the key ingredient is a data structure, working on vertex subsets of our graph, allowing two particular operations: expanding a subset and computing the symmetric difference of two stored subsets. We devise such a data structure for axis-aligned unit-square graphs, and then generalize it to any convex polygons. Our structure is based on \textit{persistent segment trees}, but to our knowledge, it has not been considered before in this form.

To sum up, we claim the following results:

\begin{itemize}
    \item There is a randomized Las Vegas algorithm which, for any graph $G = (V,E)$ of distance VC-dimension at most $d$, solves $k$-\textsc{Diameter} in $\Ot{k \cdot n^{1-1/d} \cdot m}$ time complexity (Section \ref{sec:algorithm}, Theorem \ref{thm:explicit-diameter-algorithm}); 
    \item The algorithm above can be adapted to any class of implicitly given graphs, if provided an appropriate data structure, working on the graph's neighbour lists (Section \ref{sec:implicit}, Theorem \ref{thm:implicit-diameter-algorithm});
    \item In particular, for the axis-aligned unit square intersection graphs, there is a Monte Carlo algorithm solving $k$-\textsc{Diameter} in $\Ot{k \cdot n^{7/4}}$ time complexity (Section \ref{sec:squares}, Theorem \ref{thm:polygons}a);
    \item This algorithm can be generalized to any convex polygon intersection graphs, with an additional multiplicative constant depending on the polygon's number of sides (Section \ref{sec:squares}, Theorem \ref{thm:polygons}b). 
\end{itemize}

For the general graph algorithm, the new time complexity $\Ot{k \cdot n^{1-1/d} \cdot m}$ is brought down from $\Ot{k \cdot n^{1-\frac{1}{2^d poly(d)}} \cdot m}$ previously achieved in \cite{DHV22}. This is a more practical complexity, and we (tentatively) conjecture that this bound might be a tight one for the class of $K_d$-minor free graphs, or at least for the class of graphs of distance VC-dimension bounded by $d$.

As for the paper structure, Section \ref{sec:preliminaries} of this paper introduces the most important concepts, such as (distance) VC-dimension, $\varepsilon$-nets and related theorems. In Section \ref{sec:algorithm} we introduce the main tools for constructing all the fast algorithms: the low-difference orders on graph vertices, and use them to improve the results for general sparse graphs. In Section \ref{sec:implicit} we show how to use these tools in the case of implicitly given graphs. Finally, in Section \ref{sec:squares} we apply all these concepts to achieve the original goal -- a parameterized subquadratic algorithm for unit-square graphs and then for general convex polygon intersection graphs.

\section{Preliminaries}
\label{sec:preliminaries}
\subsection{Graphs, neighbourhoods and diameters}

We assume that the reader is familiar with the notion of graphs, paths and distances. Throughout the paper, all graphs are undirected and unweighted, as well as connected. We also use the same notation for most graphs: if $G = (V,E)$ is a graph, then let $n = |V|$ and $m = |E|$.

Given a graph $G=(V,E)$ and $S\subseteq V$ let $N[S]$ denote vertices in the (closed) neighbourhood of $S$, i.e. vertices belonging to $S$ or having a neighbour in $S$. If $v\in V$ we let $N[v] = N(\{v\})$.
We also introduce the notion of $k$-neighbourhood for $k \geq 0$, denoted recursively by $N^0[S] = S,~N^k[S] = N[N^{k-1}[S]]$, which is the set of vertices with distance at most $k$ to any vertex in $S$. As stated before, the diameter of $G$ is $\diam(G) = \max_{u, v \in V} \dist(u, v)$. It is easy to see that the graph has diameter at most $k$ if and only if for every $v\in V$ we have $N^{k}[v] = V$.

\subsection{Hypergraphs and VC-dimension}
A \emph{hypergraph} is a pair $(X, \Rr)$, where $X$ is the set of \emph{vertices} and $\Rr \subseteq \mathcal{P}(X)$ is a family of subsets of $X$, the \emph{hyperedges}. Some natural examples of hypergraphs, which are most important for this paper, come from graph neighbourhoods. If $G = (V, E)$ is a graph, then:
\begin{itemize}
    \item For $k \in \mathbb{Z}$, we define $\mathcal{N}^k(G) = \{N^k[v] : v \in V\}$ as the family of all possible $k$-neighbourhoods. The hypergraph $(V, \mathcal{N}^k(G))$ is the \emph{$k$-distance hypergraph} of $G$;
    \item For the family of all balls $\mathcal{B}(G) = \bigcup_{k \geq 0} \mathcal{N}^k(G)$, we will call the hypergraph $(V, \mathcal{B}(G))$ the \emph{ball hypergraph} of $G$.
\end{itemize}

As mentioned in the introduction, the key concept needed for our results is the \emph{Vapnik– Chervonenkis dimension} \cite{VC71} of a hypergraph $(X,\Rr)$. It is defined as follows:

\begin{definition}
A hypergraph $(X,\Rr)$ \textit{shatters} a subset $Y\subseteq X$ if for every $Z\subseteq Y$ there exists $R\in \Rr$ such that $Z = R \cap Y$. In other words $|\{R \cap Y \mid R \in \Rr\}| = 2^{|Y|}$. The \emph{Vapnik-Chervonenkis dimension} (or \emph{VC-dimension}) of a hypergraph is the maximum size of a shattered subset.
\end{definition}

The following theorems recall some well-established properties of hypergraphs with bounded VC-dimension. The first one deals with VC-dimension of sub-hypergraphs and projection hypergraphs, the other one (Sauer-Shelah-Perles Lemma) bounds the number of hyperedges in terms of vertices and the VC-dimension. Most of our complexity bounds throughout the paper stem from this lemma.

\begin{theorem}
\label{thm:vcdim-basic}
Let $(X,\Rr)$ be a hypergraph where $|X| = n$, and let its VC-dimension be bounded by $d$. Then:
\begin{enumerate}
    \item If $\Rr' \subseteq \Rr$, then hypergraph $(X, \Rr')$ also has VC-dimension bounded by $d$,
    \item If $Y\subseteq X$, then hypergraph $(Y, \{Y\cap R \mid R \in \Rr\})$ also has VC-dimension bounded by $d$.
\end{enumerate}
\end{theorem}

\begin{theorem}\label{thm:ssp-lemma}\textit(Sauer-Shelah-Perles Lemma, \cite{Sa72, Sh72}). For every integer $d$, there exists a constant $\beta = \beta(d)$ such that every hypergraph $(X,\Rr)$ of VC-dimension at most $d$ satisfies $|\Rr| \leq \beta \cdot |X|^d$.
\end{theorem}

\begin{corollary}\label{cor:ssp-subset}
For every integer $d$, there exists a constant $\beta = \beta(d)$ such that for every hypergraph $(X,\Rr)$ of VC-dimension at most $d$ and for every $S \subset X$, the cardinality of $\{S \cap R \mid R \in \Rr\}$ is at most $\beta |S|^d$.
\end{corollary}

Let $(X,\Rr)$, $(X,\Rr')$ be two hypergraphs on the same underlying set $X$. Suppose that both of them have VC-dimension $d$. By $\sdiff$ we denote the \emph{symmetric difference} operator on sets, i.e. $A \sdiff B := (A \setminus B) \cup (B \setminus A)$. An important issue for us is bounding the VC-dimension of the hypergraph $(X, \{R \sdiff R' \mid R \in \Rr, R' \in \Rr'\})$. The following lemma provides a bound of $\Oh{d \log d}$. It works for any operator $\circ$ on set such that intersection distributes over $\circ$ (i.e. $A \cap (B \circ B') = (A \cap B) \circ (A \cap B')$ for any sets $A, B, B'$; the union, intersection, set difference and symmetric difference operators all have this property. It is partially based on a similar lemma in \cite{EA07}, see \ref{appendix:omitted} for more details and the proof.

\begin{restatable}{lemma}{lemmaJoinDim}
\label{lemma:join-dim}
Let $(X,\Rr)$, $(X,\Rr')$ be hypergraphs with VC-dimension not greater than $d$, and let $\circ : \mathcal{P}(X) \times \mathcal{P}(X) \to \mathcal{P}(X)$ be a binary set operator such that intersection distributes over $\circ$. Then the VC-dimension of $(X, \Rr^*)$, where $\Rr^* = \{R \circ R' \mid R \in \Rr, R' \in \Rr'\}$ is $\Oh{d \log d}$.
\end{restatable}

Let us now define another one of this paper's central concepts, linking the notion of VC-dimension with graph diameters: the \textit{distance VC-dimension} of a graph $G=(V,E)$. 

\begin{definition}
    \emph{Distance VC-dimension} of a graph $G=(V,E)$ is the VC-dimension of its ball hypergraph, i.e. the hypergraph $(V, \mathcal{B}(G))$.
\end{definition}

We assume throughout the paper that we only consider graphs with distance VC-dimension at least $2$, as there are no non-trivial connected graphs with distance VC-dimension $1$. Finally, observe that by Theorem \ref{thm:vcdim-basic}, if distance VC-dimension of a graph is bounded by some integer $d$, then for every $k \in \mathbb{N}$, the VC-dimension of $(V,\mathcal{N}^k(G))$ is also bounded by $d$.


Among others, the following classes of graphs have bounded distance VC-dimension: interval graphs, $K_t$-minor free graphs, and in general any minor-closed class of graphs \cite{BS15,CEV07}. The next section is devoted to establishing similar bounds for geometric intersection graphs.

\subsection{Geometric intersection graphs}

In this section we introduce the notion of intersection graphs and discuss their distance VC-dimension. Throughout the paper, the symbol $\oplus$ denotes the \emph{Minkowski sum} of subsets of $\mathbb{R}^2$: for any $A, B \subseteq \mathbb{R}^2$, $A \oplus B := \{(a_1 + b_1, a_2 + b_2) \mid (a_1, a_2) \in A, (b_1, b_2) \in B\}$. If $a\in \mathbb{R}^2, B \subseteq \mathbb{R}^2$ then by $a \oplus B$ we mean $\{a\} \oplus B$. We also use natural scalar multiplication: for any $\lambda \in \mathbb{R}$ and $A \subseteq \mathbb{R}^2$, $\lambda \cdot A = \{(\lambda a_1, \lambda a_2) \mid (a_1, a_2) \in A \}$.

\begin{definition}
    For a shape $\mathcal{F}\subseteq \mathbb{R}^2$, an \emph{intersection graph} $I(V,\mathcal{F})$ is a simple undirected graph with vertices $V\subseteq \mathbb{R}^2$ being points on a plane, where an edge $\{v_1,v_2\}$ for $v_1\neq v_2\in V$ exists if and only if shapes $\mathcal{F}$ centered at $v_1$ and $v_2$ have a nonempty intersection, i.e. $\left(v_1\oplus F\right) \cap \left(v_2\oplus F\right) \neq \varnothing$.
\end{definition}

Throughout this paper we assume that $\mathcal{F}$ is closed, bounded and convex.
It turns out, we can additionally assume that $\mathcal{F}$ has a center of symmetry at $(0,0)$.

\begin{restatable}{lemma}{lemmaCenterSymmetry}
    \label{lemma:symmetric}
    The graph $I(V,\mathcal{F})$ is isomorphic to $I(V, \mathcal{H})$, where $H = \frac{1}{2} \cdot\left[ \mathcal{F} \oplus (-\mathcal{F}) \right]$.
\end{restatable}
\begin{proof}
    See Appendix \ref{appendix:omitted}.
\end{proof}

Our main focus will be on the case where $\mathcal{F}$ is an $s$-sided polygon, however, even without this assumption, the geometric intersection graphs have bounded distance VC-dimension. The following lemma formally states that, and will be another crucial tool for our results.

\begin{restatable}{lemma}{lemmaBounded}
    \label{lemma:bounded}
    For any intersection graph $I(V, \mathcal{F})$, its distance VC-dimension is at most 4.
\end{restatable}

\begin{proof}
    A more general version of this lemma was elegantly proven in \cite{CGL24} using different approach, making our proof redundant. We include our proof in \ref{appendix:omitted} for the sake of completeness.
\end{proof}

Please note that the definition of $I(V,\mathcal{F})$ allows the copies of the shape $\mathcal{F}$ to be translated, but not rotated. If rotation is allowed, the lemma above does not work. Moreover, even in the case of rotated triangles, there is a construction proving a conditional quadratic lower bound for the \textsc{3-Diameter} problem as well as unbounded distance VC-dimension \cite{BKKNZ22, CGL24}.

\subsection{\texorpdfstring{$\eps$}{eps}-nets}

The concept of $\eps$-nets was introduced in \cite{HW86} and applied to bounded VC-dimension hypergraphs in \cite{CW89}. In this section, we recall the definitions and basic facts from these works, that we will need later.

\begin{definition}
    Let $(X, \Rr)$ be a hypergraph. For any $\eps >0$, a set $S \subset X$ is an \emph{$\eps$-net} if for every edge $R \in \Rr$, $|R| \geq \eps \cdot |X| \implies R \cap S \neq \varnothing$. 
\end{definition}

It turns out that this concept synergizes well with VC-dimension, as bounded VC-dimension implies any sufficiently large random subset to be an $\eps$-net with high probability:

\begin{lemma} \cite{CW89}
\label{lemma:eps-net-original}
There exists a constant $\alpha$ such that for any hypergraph $(X, \mathcal{R})$ of VC-dimension at most $d$ and for any $\delta, \varepsilon > 0$, a random set $S \subseteq X$ with $|S| \geq \alpha \cdot \frac{d}{\varepsilon} \log \frac{1}{\delta \varepsilon}$ is an $\varepsilon$-net with probability at least $1 - \delta$.
\end{lemma}

We will use this lemma with some modifications. First, we will employ a traditional notion of high probability, i.e. for a given $c$ we take $\delta = |X|^{-c}$, so the probability of failure is $\frac{1}{poly(|X|)}$. Also, as in \cite{DHV22} we employ this lemma not for the given hypergraph $(X,\mathcal{R})$, but for $(X, \mathcal{R}^*)$, where $\mathcal{R}^* = \{(R' \sdiff R'') : R', R'' \in \mathcal{R}\}$. From Lemma \ref{lemma:join-dim} we know that the VC-dimension of $\mathcal{R}^*$ is $\Oh{d \log d}$. Taking the definition of $\varepsilon$-net into account, we can reformulate Lemma \ref{lemma:eps-net-original} in the following way:

\begin{corollary}
\label{corollary:eps-net-size}
For any positive integers $c$ and $d$ there exists a constant $\alpha = \alpha(c,d)$ such that for any hypergraph $(X, \mathcal{R})$ of VC-dimension at most $d$ and for any $\varepsilon > 0$, any random set $S \subseteq X$ with $|S| \geq \alpha \cdot \frac{1}{\varepsilon} \log \frac{|X|}{\varepsilon}$ has (with probability at least $1 - |X|^{-c}$) the following property: if $R', R'' \in \mathcal{R}$ and $R' \cap S = R'' \cap S$, then $|R' \sdiff R''| \leq \varepsilon \cdot |X|$.
\end{corollary}


\section{General algorithm framework}
\label{sec:algorithm}



\subsection{Orders on hypergraphs}

The following lemma is our main tool. It provides an order of the hyperedges of any bounded distance VC-dimension hypergraph such that the difference between consecutive hyperedges is ``sufficiently small''. It corresponds to Theorem 1.2 in \cite{DHV22}, but with one important difference: in our setting ``sufficiently small'' means that the \emph{sum} of all differences is bounded, whereas in the previous work the bounds applied to \emph{every one} of the differences.

\begin{restatable}{lemma}{lemmaHypergraphOrder}
\label{lemma:hypergraph-order}
Let $(X, \Rr)$ be a hypergraph with  VC-dimension at most $d$. There exists an order $R_1, R_2, \ldots, R_{|\Rr|}$ on its hyperedges such that $\sum_{i=1}^{|\Rr|-1}|R_i \sdiff R_{i+1}| = \Oh{|\Rr|^{1-1/d} \cdot |X|}$. If there is an algorithm working in time complexity $P(|X|,|\Rr|)$ which can list, for any given $x \in X$, all $R \in \Rr$ containing $x$, then the desired order $R_1, R_2, \ldots, R_n$ can be computed, with high probability, in time complexity $\Ot{|\Rr|^{1+1/d}+|\Rr|^{1/d} P(|X|, |\Rr|})$.
\end{restatable}

\begin{proof}
Consider a weighted, undirected graph $G$ with $\Rr$ as its vertex set. For any $R', R'' \in \Rr$ we define the weight of the edge $(R', R'')$ of $G$ as $|R' \sdiff R''|$. Our immediate goal is to find a spanning tree $\mathcal{T}$ of $G$ having total cost of edges bounded by $\Oh{|\Rr|^{1-1/d} \cdot |X|}$. If we succeed, then it is easy to obtain the desired order on $\Rr$: take an Euler tour $(R_{k_{1}}, R_{k_{2}}, \ldots, R_{k_{2|\Rr|-2}})$ of $\mathcal{T}$. We know that $\sum |R_{k_{i}} \sdiff R_{k_{i+1}}|$ is also $\Oh{|\Rr|^{1-1/d} \cdot |X|}$, as each edge of $\mathcal{T}$ appears twice in an Euler tour. But if we delete some elements from the $R_{k_i}$ sequence, this value can only decrease, as $|A \sdiff C| \leq |A \sdiff B| + |B \sdiff C|$ for any finite sets $A, B, C$. Therefore we prune the sequence, keeping only the first instance of every element of $\Rr$, obtaining a path $(R_1, \ldots, R_{|\Rr|})$ with $\sum |R_j \sdiff R_{j+1}| = \Oh{|\Rr|^{1-1/d} \cdot |X|}$. Moreover, all these operations converting $\mathcal{T}$ to the order need time complexity $\Oh{|\Rr|}$. So we can now focus on finding $\mathcal{T}$.

We will use a randomized algorithm with probability of failure at most $|X|^{-c}$, for a given integer $c$. Pick a random set $S \subseteq X$ with $|S| = s = |\Rr|^{1/d}$ and arrange its elements in random order $(x_1, \ldots, x_s)$. Let $\alpha = \alpha(c+1,d)$ be the constant from Corollary \ref{corollary:eps-net-size} and fix $\varepsilon = \frac{2 \alpha \log |X|}{s}$. Now 
\[\alpha \cdot \frac{1}{\varepsilon} \log \frac{|X|}{\varepsilon} \leq \frac{s}{2 \log |X|} \cdot (\log|X| + \log s) \leq s,\] 

so $S$ and $\varepsilon$ satisfy the assumptions of Corollary \ref{corollary:eps-net-size}. Also, let $q = \lfloor \log_2 s \rfloor$ and for every $k = 0, 1, \ldots, q$ we define $s_k = \frac{s}{2^k}$ and $S_k = \{x_1, \ldots, x_{s_k}\}$. It is easy to see that $S_q \subset S_{q-1} \subset \ldots \subset S_1 \subset S_0 = S$. Observe that every prefix $S_k$ of $S$ is also a random subset of $X$, so we can also apply Corollary \ref{corollary:eps-net-size} to it if we take $\varepsilon_k = 2^k \cdot \varepsilon$ --- indeed, $\alpha \cdot \frac{1}{\varepsilon_k} \log \frac{|X|}{\varepsilon_k} \leq \alpha \cdot \frac{1}{2^k} \cdot \frac{1}{\varepsilon} \log \frac{|X|}{\varepsilon} \leq \frac{s}{2^k}$. The probability of failure for each of $S_k$ is at most $|X|^{-(c+1)}$, so by union bound, with probability greater than $1 - |X|^{-c}$ no failure will happen.

Throughout the algorithm, we maintain the partition $\Rr = \Rr_1 \cup \Rr_2 \cup \ldots \cup \Rr_t$ into disjoint subsets -- \emph{groups}. At each step, we add some edges to $\mathcal{T}$, split some of the groups into smaller parts, and maintain the invariant that $\mathcal{T}$ is a spanning tree on the set of all groups. Initially $\mathcal{T} = \varnothing$ and the partition consists of a single group $\Rr$. Now for every $j = 1, 2, \ldots, s$ we repeat the following subroutine: every group $\Rr_i$ is split into parts $\Rr_i^0 = \{R \in \Rr_i : x_j \notin R\}$ and $\Rr_i^1 = \{R \in \Rr_i : x_j \in R\}$. If both $\Rr_i^0$  and $\Rr_i^1$ are nonempty, we pick any $R_0 \in \Rr_i^0$ and $R_1 \in \Rr_i^1$, add $(R_0, R_1)$ to $\mathcal{T}$, and add both parts as new groups instead of $\Rr_i$. If one of the parts is empty, the other is $\Rr_i$, and we leave it as it is. 
After completing $s$ steps, $\mathcal{T}$ may still not span all vertices in $\mathcal{R}$. Let us call all the edges added so far the \emph{primary} edges, and then proceed to add new arbitrary edges to $\mathcal{T}$ until it becomes a tree. Those later edges we will call \emph{secondary}.

 Now consider the edges added to $\mathcal{T}$ between step $s_k$ and $s_{k-1}$ (strictly after $s_k$, but including $s_{k-1}$). For any such edge $(R',R'')$ there must be $R' \cap S_k = R'' \cap S_k$, as $(R', R'')$ belonged to the same group after step $s_k$. But as $S_k$ is an $\varepsilon_k$-net, this means that $R' \sdiff R'' \leq 2^k \cdot \varepsilon \cdot |X|$. On the other hand, let us count the number of groups before step $s_{k-1}$. For any $R'$ and $R''$ belonging to different groups there must be $R' \cap S_{k-1} \neq R'' \cap S_{k-1}$, so $R'$ and $R''$ induce two different subsets of $S_{k-1}$. But as VC-dimension of $\Rr$ does not exceed $d$, by Corollary \ref{cor:ssp-subset} there can be no more than $\beta |S_{k-1}|^d = \beta (\frac{s}{2^{k-1}})^d$ different subsets of $S_{k-1}$ induced by $\Rr$, for come constant $\beta$. This proves that before step $s_{k-1}$ there are at most $\beta(\frac{s}{2^{k-1}})^d$ edges added to $\mathcal{T}$. The cost of edges added between steps $s_k$ and $s_{k-1}$ can be bounded by $\beta (\frac{s}{2^{k-1}})^d \cdot 2^k \varepsilon |X|$, and the cost of all primary edges is at most
 \[\sum_{k=1}^{q} \beta \left( \frac{s}{2^{k-1}} \right)^d \cdot 2^k \varepsilon |X| = \Oh{q \cdot s^d \cdot \varepsilon |X|} = \Oh{q \cdot s^{d-1} \cdot |X| \cdot \log |X|} = \Ot{|\Rr|^{1-1/d} \cdot |X|}.\]
For the secondary edges, the bound is even simpler: for every such edge $(R',R'')$ we already know that $R'$ and $R''$ ended up in the same group, so $R' \cap S = R'' \cap S$, which means $|R' \sdiff R''| \leq \varepsilon |X|$. As there are at most $|\Rr|$ such edges, we bound their cost by \[|\Rr| \cdot \varepsilon |X| = \Oh{|\Rr| \cdot \frac{\log |X|}{|\Rr|^{1/d}} \cdot |X|} = \Ot{|\Rr|^{1-1/d} \cdot |X|},\]
which completes the proof.
As for the complexity, a single step in the first phase (for primary edges) takes $\Oh{|\Rr| + P(|X|,|\Rr|)}$ time, as it has to go, for a fixed $x_j$, through all $R \in \Rr$ and determine if $x_j \in R$. Adding secondary edges is $\Oh{|\Rr|}$. Therefore, the total complexity is $\Ot{|\Rr|^{1+1/d}+|\Rr|^{1/d} P(|X|, |\Rr|})$.

\end{proof}
\noindent We will apply Lemma \ref{lemma:hypergraph-order} to distance hypergraphs, using its two variants: 
\begin{itemize}
\item The first one (Corollary \ref{corollary:distance-hypergraph-xors}) uses $k$-neighbourhoods as the edges of the hypergraph, so in the resulting order the adjacent vertices have similar $k$-neighbourhoods. We can directly compute the next neighbourhood from the previous one. This will mainly be useful for geometric intersection graphs and other implicitly given graphs.
\item The second one (Corollary \ref{corollary:weighted-average-spanning-path}) uses duality (reverses the role of vertices and their $k$-neighbourhoods), which allows us to express every $k$-neighbourhood as a sum of sublinear number of intervals. This will be useful for the general sparse graph case.
\end{itemize}

\begin{restatable}{corollary}{corDistanceHypergraphXors}
\label{corollary:distance-hypergraph-xors}
Let $G = (V,E)$ be a graph with distance VC-dimension at most $d$ and let $k \in \{1, 2, \ldots, n\}$ There is an order $v_1, \ldots, v_n$ on the vertices of $G$ such that $\sum_{i=1}^{n-1} |N^k[v_i] \sdiff N^k[v_{i+1}]| = \Oh{n^{2-1/d}}$. This order can be computed, with high probability, in time complexity $\Ot{n^{1/d} \cdot T(G)}$, where $T(G)$ is the complexity of a single-source distance finding algorithm (e.g. BFS).
\end{restatable}
\begin{proof}
    See Appendix \ref{appendix:omitted}.
\end{proof}

Let $G = (V,E)$ be a graph, and let $\sigma = (x_1, \ldots, x_n)$ be some order on vertex set $V$. For any $a, b \in \{1,2, \ldots, n\}$, $a \leq b$ let  $x[a,b]$ be some interval of vertices in this order, i.e. $x[a,b] = \{x_a, x_{a+1}, \ldots, x_b\}$. Every subset $D \subset V$ can be expressed as the sum of such intervals: $D = \bigcup_{i=1}^{s} x[a_i, b_i]$ for some positive integer $s$, and some $a_1, \ldots, a_s, b_1, \ldots, b_s$. There exists exactly one such representation with minimal possible $s$ -- let us call it $I_\sigma(D)$, the \emph{canonical interval representation of $D$} with respect to the order $\sigma$. We will omit $\sigma$ and write $I(D)$ whenever it is clear from the context.

\begin{restatable}{corollary}{corWeightedAverageSpanningPath}
\label{corollary:weighted-average-spanning-path}
Let $G = (V,E)$ be a graph with distance VC-dimension at most $d$, let $k \in \{1, 2, \ldots, n\}$, and let $\alpha : V \to \mathbb{Z}^+$ be any assignment of positive integer weights to vertices. There exists an order $v_1, \ldots, v_n$ on the vertices of $G$ such that, with respect to that order, $\sum_{x \in V} \alpha(x) \cdot |I(N^k[x])| = \Ot{n^{1 - 1/d} \cdot \sum_{j=1}^{n} \alpha(j)}$. This order can be computed, with high probability, in $\Ot{n^{1/d} \cdot T(G)}$ time complexity, where $T(G)$ is the complexity of a single-source distance finding algorithm (e.g. BFS).
\end{restatable}

\begin{proof}
    See Appendix \ref{appendix:omitted}.
\end{proof}

\subsection{Algorithm for general sparse graphs}

We now prove our main result for general graphs of bounded distance VC-dimension.

\begin{theorem}
	\label{thm:explicit-diameter-algorithm}
	Let $\mathcal{G}$ be a graph class with distance VC-dimension bounded by $d \geq 2$.
	There exists an algorithm that decides if a graph $G \in \mathcal{G}$ has diameter at most $k$ in $\Ot{kmn^{1-1/d}}$ time with high probability.
\end{theorem}

Our algorithm builds on the work of Ducoffe et al. \cite{DHV20, DHV22}, whose algorithm iteratively computes $r$-neighbourhoods $N^r[v]$ for all $v \in V$ and $r \in \{0, \ldots, k\}$. Note that each set $N^r[v]$ can have $\Oh{n}$ elements, so their total size can be $\Oh{kn^2}$. This means that we cannot break the quadratic barrier if we store the vertex sets explicitly.

To alleviate this, \cite{DHV22} uses \emph{spanning paths with low stabbing number}, i.e. arranges vertices in a particular order such that $\max_{x \in V} |I(N^r[x])| = \Ot{n^{1-\eps_d}}$ for a fixed $r$. The authors provide a subquadratic algorithm that finds such an order with $\eps_d$ dependent only on distance VC-dimension $d$.
Then they use interval representations to encode and operate on the $r$-neighbourhoods.
This yields an algorithm with running time $\Ot{kmn^{1-\eps_d}}$.

One of our goals is to improve the constant $\eps_d$. It is known that there always exists a vertex order with $\eps_d = 1/d$ (\cite{CW89}), but there is no known algorithm that computes it in subquadratic time. Instead, we observe that we can relax requirements for the vertex orders.
Namely, it is sufficient to obtain low weighted average instead of maximum over interval representations. This enables us to use the algorithm given by Corollary \ref{corollary:weighted-average-spanning-path}.

\subparagraph*{Ball encoding.} As we assumed $G$ to be connected and non-trivial, we know that $\deg(v) \geq 1$ for each $v \in V$. Let $v^r_1, \ldots, v^r_n$ be a vertex order such that with respect to that order the following holds:
\begin{equation*}
	\sum_{x \in V} \deg(x) \cdot |I(N^r[x])| = \Ot{mn^{1 - 1/d}}.
\end{equation*}
From Corollary \ref{corollary:weighted-average-spanning-path}, using vertex weights $\alpha(v) = \deg(v) \geq 1$, we know that such an order exists and can be computed (with high probability) in $\Ot{n^{1/d} m}$ time complexity.
Our algorithm encodes $r$-neighbourhoods using their canonical interval representations with respect to $v^r_1, \ldots, v^r_n$.
More precisely, we compute sets of intervals $\mathcal{I}^r_v = I(N^r[v])$ for all vertices $v \in V$.

\begin{lemma}
	\label{lem:explicit-algorithm-step}
	Suppose we are given the encoding for $(r-1)$-neighbourhoods, i.e. the vertex order $v^{r-1}_1, \ldots, v^{r-1}_n$ and the representations $\mathcal{I}^{r-1}_v$ for all vertices $v \in V$.
	Then we can compute the encoding for $r$-neighbourhoods in $\Ot{mn^{1-1/d}}$ time with high probability.
\end{lemma}
\begin{proof}
	The algorithm proceeds as follows.
	\begin{enumerate}
		\item
			For each vertex $v \in V$, compute the interval representation $\mathcal{I}'_v$ of $N^r[v]$ with respect to the old vertex order $v^{r-1}_1, \ldots, v^{r-1}_n$.
			Note that $N^r[v] = \bigcup_{x \in N[v]} N^{r-1}[x]$.
			This means that we can compute $\mathcal{I}'_v$ by summing the representations $\mathcal{I}^{r-1}_x$ over neighbours $x \in N[v]$.
			This can be done using a standard line sweep procedure in $\Ot{\sum_{x \in N[v]}{|\mathcal{I}^{r-1}_x|}}$.
			Overall this step takes time
			$\Ot{\sum_{v \in V} \sum_{x \in N[v]}{|\mathcal{I}^{r-1}_x|}} =
				\Ot{\sum_{x \in V} \deg(x) \cdot |\mathcal{I}^{r-1}_x|} = \Ot{mn^{1 - 1/d}}$.
			The total size of all representations $\mathcal{I}'_v$ is $\Ot{mn^{1 - 1/d}}$ as well.
		\item
			Compute the new vertex order $v^r_1, \ldots, v^r_n$ via Corollary \ref{corollary:weighted-average-spanning-path}.
			To achieve this, we only need to provide an algorithm that lists vertices of $N^r[v]$ efficiently for a given vertex $v \in V$.
			This can be implemented easily in $\Oh{m}$ time using breadth-first search.
			It follows that the vertex order can be computed in $\Ot{mn^{1/d}}$ time.
		\item
			Compute the canonical interval representations $\mathcal{I}^r_v$ with respect to the new vertex order $v^r_1, \ldots, v^r_n$.
			We do this by transforming the representations $\mathcal{I}'_v$ as follows.
			\begin{enumerate}
				\item
					Let $A_i$ be the set of vertices $x \in V$ such that $v^r_i$ is the left endpoint of an interval in $\mathcal{I}^r_x$.
					Consider a vertex $x \in A_i$ for $i \geq 2$.
					We have that $v^r_i \in N^r[x]$ and $v^r_{i-1} \notin N^r[x]$.
					This is equivalent to $x \in N^r[v^r_i]$ and $x \notin N^r[v^r_{i-1}]$.
					It follows that $A_i = N^r[v^r_i] \setminus N^r[v^r_{i-1}]$ for $i \geq 2$.
					We can thus compute $A_i$ from interval representations $\mathcal{I}'_{v^r_{i-1}}$ and $\mathcal{I}'_{v^r_i}$ in $\Ot{|\mathcal{I}'_{v^r_{i-1}}| + |\mathcal{I}'_{v^r_i}| + |A_i|}$ time using line sweep procedure.
					It remains to handle the case when $i = 1$.
					By similar argument, we obtain that $A_1 = N^r[v^r_1]$, so it is enough to list vertices in $\mathcal{I}'_{v^r_1}$.
					Overall, this step takes time $\Ot{\sum_{v \in V} |\mathcal{I}'_v| + \sum_{i=1}^n |A_i|} = \Ot{mn^{1-1/d}}$.
				\item
					Let $B_i$ be the set of vertices $x \in V$ such that $v^r_i$ is the right endpoint of an interval in $\mathcal{I}^r_x$.
					We can compute all these sets in time $\Ot{mn^{1-1/d}}$, similarly as above.
				\item
					Recover the interval representations $\mathcal{I}^r_v$ for all $v \in V$ from the sets $A_1, \ldots A_n$ and $B_1, \ldots, B_n$.
					This step takes time $\sum_{i=1}^n |A_i|+|B_i| = \Ot{mn^{1-1/d}}$.
			\end{enumerate}
	\end{enumerate}
	The total running time is $\Ot{mn^{1/d} + mn^{1-1/d}}$, which becomes $\Ot{mn^{1-1/d}}$ for $d \geq 2$.
\end{proof}

\begin{proof}[Proof of Theorem \ref{thm:explicit-diameter-algorithm}]
	We start with an arbitrary permutation of vertices $v^0_1, \ldots, v^0_n$,
	and trivial interval representation $\mathcal{I}^0_v = I(\{v\})$ for each vertex $v \in V$.
	Then we compute the encodings of all $k$-neighbourhoods inductively using Lemma $\ref{lem:explicit-algorithm-step}$.
	Finally, we check if there a vertex $v \in V$ such that $\mathcal{I}^k_v \neq I(V)$.
	If that is the case, the diameter is larger than $k$.
	Otherwise it is at most $k$.
\end{proof}

\section{Diameter testing for implicit graphs}
\label{sec:implicit}

In this section, we consider the diameter problem for graphs of bounded distance VC-dimension that admit implicit representations.
We propose a diameter testing algorithm that relies on the existence of a certain data structure, but is independent of the number of edges.
In particular, this framework can be applied for geometric intersection graphs. In Section \ref{sec:squares} we show an implementation for unit squares. Please refer to Appendix \ref{appendix:omitted} for a generalization for arbitrary convex polygons.

We begin by introducing a necessary data structure template. The \emph{Neighbour Set Data Structure} (\emph{NSDS} for short) maintains a family $\mathcal{T}$ of vertex subsets of a graph $G$ under the following operations:
\begin{itemize}
	\item $\widetilde{S'} \gets \AddNeighbours(\widetilde{S}, v)$: Given a vertex subset $\widetilde{S} \in \mathcal{T}$ and a vertex $v \in V(G)$, add a new set $\widetilde{S'} = \widetilde{S} \cup N_G[v]$ to the family $\mathcal{T}$.
	\item $D \gets \ListDifferences(\widetilde{S}_1, \widetilde{S}_2)$: Given vertex subsets $\widetilde{S}_1, \widetilde{S}_2 \in \mathcal{T}$, output their symmetric difference $D = \widetilde{S}_1 \sdiff \widetilde{S}_2$.
\end{itemize}
Initially, the family $\mathcal{T}$ contains only the empty vertex set $\varnothing$.
Throughout the whole section, we will use a tilde to mark vertex sets registered within NSDS (e.g. $\widetilde{S}$). Such vertex sets are represented implicitly by references to the data structure, so the time complexity of some operations on them may be a lot smaller than their size.

We say that a graph class $\mathcal{G}$ admits an \emph{efficient implementation} of Neighbour Set Data Structure if the operations can be implemented in the following time complexities:
\begin{itemize}
	\item initialization in $\Ot{n}$ time (given an implicit $\Oh{n}$-size representation of a graph $G \in \mathcal{G}$);
	\item $\AddNeighbours$ in $\Ot{1}$ time;
	\item $\ListDifferences$ in $\Ot{|D|}$ time.
\end{itemize}

The remainder of this section is devoted to proving the following theorem.

\begin{theorem}
	\label{thm:implicit-diameter-algorithm}
	Let $\mathcal{G}$ be a graph class with distance VC-dimension bounded by $d \geq 2$.
	If $\mathcal{G}$ admits an efficient implementation of Neighbour Set Data Structure,
	then there exists an algorithm that decides if a graph $G \in \mathcal{G}$ has diameter at most $k$ in time $\Ot{kn^{2-1/d}}$ with high probability.
\end{theorem}

\subsection{Algorithm outline}
\label{sec:implicit-diameter-algorithm-outline}

In this section we describe the high-level idea of the algorithm from Theorem \ref{thm:implicit-diameter-algorithm}, leaving some subprocedures and other technical details to following subsections.
The algorithm iteratively computes $r$-neighbourhoods $N^r[v]$ for all vertices $v \in V$, but in a different way than in Section $\ref{sec:algorithm}$.
In particular, a different encoding of neighbourhoods is used.

\subparagraph*{Balls encoding.} 

Let $v^r_1, \ldots, v^r_n$ be the vertex order for the $r$-neighbourhoods produced by Corollary \ref{corollary:distance-hypergraph-xors}. Observe that to apply this corollary, we only need a single-source shortest path finding algorithm, and we show that the classical BFS algorithm can be simulated using NSDS in $\Ot{n}$ time (see Appendix \ref{appendix:omitted} for details). Using this vertex order, the $r$-balls are now delta-encoded using vertex sets $D^r_1, \ldots, D^r_n$ such that:
\begin{align*}
	N^r[v_i^r] = D^r_1 \sdiff \ldots \sdiff D^r_i &&
	D^r_i =
		\begin{cases}
			N^r[v_1] & \text{for $i = 1$} \\
			N^r[v_{i-1}] \sdiff N^r[v_i] & \text{for $i \in \{2, \ldots, n\}$}
		\end{cases}
\end{align*}
It immediately follows from Corollary \ref{corollary:distance-hypergraph-xors} that total size of all sets $D^r_1, \ldots, D^r_n$ for a fixed $r$ is bounded by $\Ot{n^{2-1/d}}$.
Note that this is essentially a transposition of the encoding used in Section $\ref{sec:algorithm}$, where each ball was represented by a set of intervals. Here, each set $D_i$ is in fact the set of these balls which have one of their intervals ending between $i-1$ and $i$.

\subparagraph*{Algorithm step.}
Suppose we have already computed the encoding for $(r-1)$-balls, i.e. the vertex order $v^{r-1}_1, \ldots, v^{r-1}_n$ and the sets $D^{r-1}_1, \ldots, D^{r-1}_n$.
To compute the encoding for $r$-balls, we proceed as follows.
\begin{enumerate}
	\item
		Build representations $\widetilde{B}^r_1, \ldots, \widetilde{B}^r_n$ of $r$-balls in the NSDS.
		Specifically, we want $\widetilde{B}^r_i = N^r[v^{r-1}_i]$, i.e. we still use the vertex order for $(r-1)$-balls.
		Observe that $N^r[v] = N[N^{r-1}[v]]$.
		To naively compute the set $\widetilde{B}^r_i$, one could invoke $\AddNeighbours$ for each vertex $v \in N^{r-1}[v^{r-1}_i]$.
		Clearly, such approach would be too slow.
		In Subsection \ref{sec:balls-expansion}, we provide a divide-and-conquer algorithm that builds all the representations efficiently in $\Ot{n + \sum_i{|D^{r-1}_i|}} = \Ot{n^{2-1/d}}$ time.
	\item
		Compute new vertex order $v^r_1, \ldots, v^r_n$ using Corollary \ref{corollary:distance-hypergraph-xors} in $\Ot{n^{1+1/d}}$ time.
	\item
		Let $\pi_r$ be a permutation mapping vertex indices for $r$-balls into vertex indices for $(r-1)$-balls, i.e. $\pi_r(i) = j$ if and only if $v_i^r = v_j^{r-1}$.
	\item
		Compute new delta-encoding $D^r_1, \ldots, D^r_n$.
		By definition, $D_r^i = N^r[v^r_{i-1}] \sdiff N^r[v^r_i] = \widetilde{B}^r_{\pi_r(i-1)} \sdiff \widetilde{B}^r_{\pi_r(i)}$ for $i \geq 2$.
		This means that we can compute each set $D_r^i$ by invoking \linebreak $\ListDifferences(\widetilde{B}^r_{\pi_r(i-1)}, \widetilde{B}^r_{\pi_r(i)})$ on computed representations.
		To compute $D^r_1$ we can use $\ListDifferences(\varnothing, \widetilde{B}^r_{\pi_r(1)})$.
		Since $\ListDifferences$ operation is output-sensitive, the overall complexity of this step is $\Ot{n + \sum_i{|D^r_i|}} = \Ot{n^{2-1/d}}$.
\end{enumerate}
The total runtime of a single step is $\Ot{n^{1+1/d} + n^{2-1/d}}$, which for $d \geq 2$ becomes $\Ot{n^{2-1/d}}$.

\subparagraph*{Full algorithm.}
We start with an arbitrary permutation of vertices $v^0_1, \ldots, v^0_n$.
Moreover, we have $D^0_1 = N^0[v^0_1] = \{v^0_1\}$ and $D^0_i = N^0[v^0_{i-1}] \sdiff N^0[v^0_i] = \{v^0_{i-1}, v^0_i\}$ for $i \geq 2$.
Then we compute the encodings of all $k$-neighbourhoods by repeatedly applying the algorithm step.
Finally, we check if $D^k_1 = V$ and $D^k_i = \varnothing$ for all $i \geq 2$.
If that is the case, the diameter is at most $k$.
Otherwise, it is larger than $k$.
We obtain the final time complexity $\Ot{kn^{2-1/d}}$.


\subsection{Ball expansion}
\label{sec:balls-expansion}

We now describe an algorithm that builds representations of $r$-balls in NSDS given a delta-encoding of $(r-1)$-balls.
Specifically, we are given vertex sets $D^{r-1}_1, \ldots, D^{r-1}_n \subseteq V$ and we need to compute representations $\widetilde{B}^r_1, \ldots, \widetilde{B}^r_n$ such that $\widetilde{B}^r_i = N[D^{r-1}_1 \sdiff \ldots \sdiff D^{r-1}_i]$.
A naive approach would be to build all representations separately.
Instead we use a divide-and-conquer scheme that enables us to share common parts between the computed representations.

The recursive procedure takes as input vertex sets $D_1, \ldots, D_t \subseteq V$ and a data structure representation $\widetilde{S}$.
The output of the procedure are representations $\widetilde{B}_1, \ldots, \widetilde{B}_t$ such that $\widetilde{B}_i = \widetilde{S} \cup N[D_1 \sdiff \ldots \sdiff D_i]$.
We set $\widetilde{S} = \varnothing$ for the initial call; it is used later for recursion.

We begin with reduction of common vertices.
Let $C = \bigcup_{i=2}^t D_i$ and consider a vertex $v \in D_1 \setminus C$.
The vertex $v$ appears in all vertex sets of form $D_1 \sdiff \ldots \sdiff D_i$.
This means that $N[v] \subseteq \widetilde{B}_i$ for all $i \in \{1, \ldots, n\}$.
We can thus update $\widetilde{S}$ with $N[v]$ by invoking $\AddNeighbours(\widetilde{S}, v)$ and remove $v$ from $D_1$ without changing the output.
We do this for all vertices $v \in D_1 \setminus C$.
Let $\widetilde{S}'$ be the updated representation $\widetilde{S}$ and $D_1' = D_1 \cap C$ be the reduced set $D_1$.

If $t = 1$ then we are done: we just return the updated $\widetilde{S}' = \widetilde{B}_1$.
Otherwise, we use recursion.
Let $m = \floor{t/2} + 1$.
We split the sequence into halves $D_1, \ldots, D_{m-1}$ and $D_m, \ldots, D_t$.
Computing the representations $\widetilde{B}_1, \ldots, \widetilde{B}_{m-1}$ is straightforward: we recurse with $\widetilde{S}'$ and $D'_1, D_2, \ldots, D_{m-1}$.
Then $\widetilde{B}_i = \widetilde{S}' \cup N[D_1' \sdiff D_2 \sdiff \ldots \sdiff D_i] = \widetilde{S} \cup N[D_1 \sdiff \ldots \sdiff D_i]$.

To compute $\widetilde{B}_m, \ldots, \widetilde{B}_t$, we need to take into account the sets $D_1, \ldots, D_{m-1}$.
We do this by replacing $D_m$ with $D'_m = D_1' \sdiff D_2 \sdiff \ldots \sdiff D_m$.
We recurse with $\widetilde{S}'$ and $D'_m, D_{m+1}, \ldots, D_t$.
Then $\widetilde{B}_i = \widetilde{S}' \cup N[D_m' \sdiff D_{m+1} \sdiff \ldots \sdiff D_i] = \widetilde{S}' \cup N[D_1' \sdiff D_2 \sdiff \ldots \sdiff D_i] = \widetilde{S} \cup N[D_1 \sdiff \ldots \sdiff D_i]$. This completes the description of the procedure.

We provide the pseudocode as Algorithm \ref{alg:balls-expansion}.

\begin{algorithm}[H]
	\caption{Balls expansion procedure.}
	\label{alg:balls-expansion}
	\begin{algorithmic}[1]
		\Function{ExpandBalls}{$D_1, \ldots, D_t$}
			\State \Return $\ExpandBallsRecursively(\varnothing, D_1, \ldots, D_t)$
		\EndFunction
		\State
		\Function{ExpandBallsRecursively}{$\widetilde{S}, D_1, \ldots, D_t$}
			\State $\widetilde{S}' \gets \widetilde{S}$, $C \gets \bigcup_{i=2}^t D_i$
			\ForAll{$v \in D_1 \setminus C$}
				\State $\widetilde{S}' \gets \AddNeighbours(\widetilde{S}', v)$
			\EndFor
			\If{t = 1}
				\State \Return $\widetilde{S}'$
			\EndIf
			\State $m \gets \floor{t/2} + 1$
			\State $D_1' \gets D_1 \cap C$
			\State $D_m' \gets D_1' \sdiff D_2 \sdiff \ldots \sdiff D_m$
			\State $\widetilde{B}_1, \ldots, \widetilde{B}_{m-1} \gets \ExpandBallsRecursively(\widetilde{S}', D_1', D_2, \ldots, D_{m-1})$
			\State $\widetilde{B}_m, \ldots, \widetilde{B}_t \gets \ExpandBallsRecursively(\widetilde{S}', D_m', D_{m+1}, \ldots, D_t)$
			\State \Return $\widetilde{B}_1, \ldots, \widetilde{B}_t$
		\EndFunction
	\end{algorithmic}
\end{algorithm}

\begin{restatable}{lemma}{lemmaExpandBalls}
	The $\ExpandBalls$ procedure works in time $\Ot{t + \sum_{i=1}^{t} |D_i|}$.
\end{restatable}
\begin{proof}
    See Appendix \ref{appendix:omitted}.
\end{proof}

\section{Polygon intersection graphs}
\label{sec:squares}
Here we give a concrete application of our framework to geometric intersection graphs. We start with stating our main theorem:

\begin{theorem}\label{thm:polygons}
There is a Monte Carlo algorithm solving the $k$\textsc{-Diameter} problem for the class of intersection graphs $I(V,\mathcal{F})$, where $\mathcal{F}$ is
    \begin{enumerate}
    \item[a)] a unit square, in $\Ot{k \cdot n^{\frac{7}{4}}}$ time.
    \item[b)] a convex $s$-sided polygon, in $\Ot{k \cdot n^{\frac{7}{4}}}$time, with a constant factor dependent on $s$.
    \end{enumerate}
\end{theorem}

By Lemma \ref{lemma:symmetric} we can assume that polygon $\mathcal{F}$ is centrally symmetric with center of symmetry at $(0,0)$. Recall that all intersection graphs of convex shapes have their distance VC-dimension bounded by 4 (Lemma \ref{lemma:bounded}). Hence, to apply Theorem \ref{thm:implicit-diameter-algorithm} and complete the proof, we just need to supply a proper Neighbour Set Data Structure. This would result in an $\Ot{k\cdot n^{2-1/d}} = \Ot{k \cdot n^{7/4}}$ time algorithm for the $k$\textsc{-Diameter} problem. The rest of this section gives a (very) rough sketch of such a data structure, focusing mainly on the unit-square case. For a more detailed description as well as the generalization to convex polygons, see Appendix \ref{appendix:squares}.

First, observe that the neighbourhood of a vertex (point) $v$ in a unit-square intersection graph $I(V,\square)$ consists simply of all points inside a square of side $2$ centered at $v$. This stays true for any graph $I(V,\mathcal{F})$:

\begin{restatable}{observation}{obsDoubling}
    \label{obs:doubling}
    Vertex $u$ is a neighbour of vertex $v$ in $I(V,\mathcal{F})$ if and only if $u\in v + 2\mathcal{F}$.
\end{restatable}

\begin{proof}
    See Appendix \ref{appendix:squares}.
\end{proof}

Therefore, if we scale all points in $V$ by a factor of 2, any neighbourhood is a simple $\mathcal{F}$-shape centered at a point. For the intersection graphs, we now assume that the desired NSDS stores some family $\mathcal{T}$ of subsets of $V$ (i.e. sets of points). We can now reformulate its operations in the following way:

\begin{itemize}
    \item $\Mark(\widetilde{S}, (x, y))$: Given a set $\widetilde{S} \in \mathcal{T}$ and a point $(x,y) \in \mathbb{R}^2$, add a new set $\widetilde{S'} = \widetilde{S} \cup P$ to the family $\mathcal{T}$, where $P \subseteq V$ contains the points covered by $\mathcal{F}$ centered at the point $(x,y)$.
    \item $\ListDifferences(\widetilde{S}_1, \widetilde{S}_2)$: Given sets $\widetilde{S}_1, \widetilde{S}_2 \in \mathcal{T}$, output their symmetric difference $D = \widetilde{S}_1 \sdiff \widetilde{S}_2$.
\end{itemize}

To fulfill the assumptions of Theorem \ref{thm:implicit-diameter-algorithm}, the $\Mark$ operation should work in $\Ot{1}$ time complexity, and $\ListDifferences$ in $\Ot{|D|}$ time complexity. We also allow initialization in $\Ot{n}$ time, where $n = |V|$ is the number of points.

\begin{restatable}{lemma}{lemmaStructureGeo}
\label{lemma:structure-geo}
The following holds true:
\begin{enumerate}
\item[a)] There exists an efficient implementation of Neighbouring Set Data Structure for the unit-square intersection graphs.
\item[b)] Let $s \in \mathbb{N}_+$ be a constant and $\mathcal{F}$ be a convex $s$-sided polygon with a center of symmetry. There exists an efficient implementation of Neighbouring Set Data Structure for the intersection graphs $I(V,\mathcal{F})$.
\end{enumerate}
\end{restatable}

To complete this section, we provide a high-level overview of the proof of Lemma \ref{lemma:structure-geo}a. For more details and the general version of NSDS, see Appendix \ref{appendix:squares}.

We start by dividing the plane into horizontal strips of height $1$ and focusing only on one such strip -- let $V$ be the set of points in the strip. To store subsets of $V$, we use a data structure called \emph{persistent segment tree}. 

\subparagraph*{Segment trees.} Let $V = \{v_1, v_2, \ldots, v_k\}$, and assume that the points are sorted by their $x$ coordinate (we can make sure of it during initialization). We can also assume that $k$ is a power of $2$, adding dummy points if needed. Let $V[i,j]$ denote the set $\{v_i, v_{i+1}, \ldots, v_{j}\}$ for any $1 \leq i \leq j \leq k.$ A \emph{segment tree} is a complete binary tree in which every node stores information associated with some interval of points $V[i,j]$. The root corresponds to $V[1,k] = V$ and any node associated with interval $V[i,j]$ with $i<j$ has two children corresponding to $V[i,s]$ and $[s+1,j]$, where $s = \lfloor \frac{(i+j)}{2} \rfloor$. The leaves of the tree correspond to single-element intervals. The height of this tree is clearly $\Oh{\log k}$.

A single instance of a tree stores a particular subset $\widetilde{S} \subseteq V$ in the following way: in every node $z$ associated with an interval $[i,j]$ we keep the subset $\widetilde{S} \cap V[i,j]$. We want, however, to minimize the stored information, and instead of the whole subset $\widetilde{S} \cap V[i,j]$ we will only remember one integer -- the \emph{hash} of this subset. Formally, with every element $v \in V$ we associate a random integer (hash) $h(v)$. For every node $z$ the subset $A_z$ stored in this node is replaced by $\bigoplus_{v \in A_x} h(v)$, i.e. the bitwise-XOR of its elements' hashes. 

Now, we must find a way to store multiple distinct sets $\widetilde{S}$, and we will achieve that by employing \emph{persistency}.

\subparagraph*{Persistency and the \ListDifferences{} operation.} Suppose that our tree currently stores a set $\widetilde{S}$, and we want to create and store a new set $\widetilde{S'} = \widetilde{S} \cup \{v_i\}$ by adding a single element. This change requires modifying the subset (i.e. its hash) in the node responsible for $V[i,i]$ and then going up along the path to the root, correcting the subsets in $\log_2 k$ nodes. Instead of modifying these nodes in-place, we employ a standard path copying technique. The nodes are immutable and copied whenever they are updated, with children links adjusted accordingly. In particular, a new copy of the root node will be created, and this copy will correspond to the new set $\widetilde{S'}$. Observe that this allows us, for every set $\widetilde{S}$ ever created, to reconstruct its subset stored in every node. This enables a relatively simple implementation of \ListDifferences($\widetilde{S}$, $\widetilde{S'}$): we start in the root and compare the hashes of $\widetilde{S}$ and $\widetilde{S'}$ in all nodes we visit. For a node $z$ associated with interval $[i,j]$ we compare hashes of $\widetilde{S} \cap V[i,j]$ and $\widetilde{S'} \cap V[i,j]$. If equal, then these subsets are equal with high probability. If not, there is at least one difference between $\widetilde{S}$ and $\widetilde{S'}$ on the interval $[i,j]$, and we recurse on both children of $z$. If $z$ is a leaf, than the difference $\widetilde{S} \sdiff \widetilde{S'}$ is the single element of $z$. An easy analysis shows that the \ListDifferences{} works in $\Oh{|D| \cdot \log k}$ time, where $D$ is the output set -- as required. However, we have only considered simple modifications of subsets (adding one element) and our desired \Mark{} operation needs way more.

\subparagraph*{More node data and the \Mark{} operation.} Recall that we work on a single strip of height $1$. Now we want to be able to modify some subset $\widetilde{S}$ by adding to it a whole unit square centered at some point $[x,y]$ (which does not have to belong to $\widetilde{S}$, and can even lay outside of our strip, having only some part of the square inside). To achieve that, we need more information stored in every node of the tree. Recall that the points are sorted according to their $x$ coordinate, so every node corresponds to some connected part of our strip. 

\begin{figure}
    \begin{center}
        \includegraphics[scale=0.75]{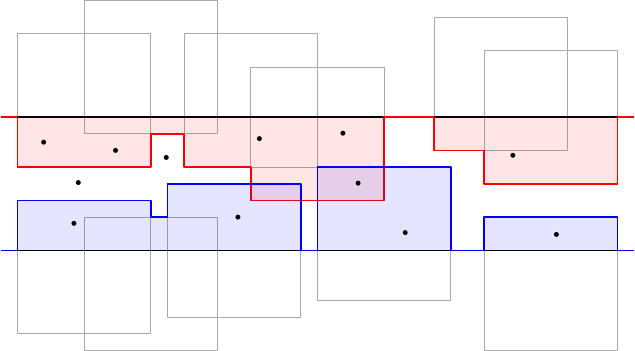}
    \end{center}
    \caption{The top and bottom areas in a node.}
    \label{fig:borders-recap}
\end{figure}

A node starts with an empty set and then more and more points become marked, all points coming from some unit squares. Each square crosses either top or bottom end of the stripe, so let us call the union of the top/bottom squares the \emph{top/bottom area}, respectively (see Figure \ref{fig:borders-recap}, note that the top and bottom areas do not have to be disjoint). Now we keep the hashes of top and bottom areas separately, and the \Mark{} operation hinges on the following observations:
\begin{itemize}
\item If the top and bottom areas are disjoint, then the hash of node's subset can be computed from the hashes of top and bottom areas;
\item If a top square is added and this square covers the node's whole top area, it is easy to update the hash of the top area; the identical fact holds for bottom squares;
\item If the top and bottom area together cover all the node's points, then the hash is trivial;
\item Any situation not falling into above categories happens relatively rarely and adds little to the time complexity of \Mark{}.
\end{itemize}

For the detailed analysis of the \Mark{} operation, see Appendix \ref{appendix:squares}.

\subparagraph*{Joining the stripes.} Finally, we need to gather the information from all the stripes. Observe that any \Mark{} only affects at most 2 stripes, so it will have the same complexity. But \ListDifferences{} is harder -- we are given sets $\widetilde{S}$ and $\widetilde{S'}$ and need to process only the stripes on which marked points are different in these two sets. To do this, we use another persistent segment tree, but with whole stripes as elements of its underlying set. For every stripe we keep the hash of its marked points, which allows us to identify the differing stripes in time complexity proportional to the number of such stripes. Then we invoke \ListDifferences{} on them.

\section{Conclusion and open problems}
\label{sec:conclusion}
\textbf{General Diameter problem.} The algorithms presented above solve $k$-\textsc{Diameter} in subquadratic time, but not general \textsc{Diameter} problem. Therefore, the first important open question is:
\begin{problem}
What is the complexity of \textsc{Diameter} for geometric intersection graphs?
\end{problem}

The paper \cite{DHV22}, on which we based our main algorithm, also provides an algorithm solving \textsc{Diameter} for $K_t$-minor free graphs, in $\Ot{n^{2-\varepsilon_t}}$ time complexity. But unlike $k$-\textsc{Diameter} case, this algorithm does not seem to be easily translated to geometric intersection setting using our technique. We would like to briefly discuss here the obstacles we encountered, as well as some related results.

The algorithm from \cite{DHV22} relies on the notion of \emph{separators}, which it uses to construct \emph{$r$-divisions}. For a graph $G = (V,E)$ with $|V| = n$, a \emph{separator} is a subset of $V$ which has a sublinear size, and removing it would split $G$ into connected components no larger than $\frac{2}{3}n$. An $r$-division is (very roughly speaking) a subset of vertices which also has sublinear size, and splits the graph into clusters of size no larger than $r$.

The core idea of the algorithm is to use a single-source path-finding algorithm (e.g. BFS), on vertices from the $r$-division of $G$, computing all the neighbourhoods of these vertices. This proves sufficient to determine all the other neighbourhoods.

Why does this algorithm not work in our case? It is even more surprising considering the existence of strong results involving separators in geometric intersection graphs (see \cite{BBK20}). However, these separators have an important difference: they are not of sublinear size by themselves, but rather can be expressed as a union of a small number of cliques. Thus we cannot run a BFS from every vertex in such a separator. It is, however, possible to use a multi-source path-finding algorithm, starting from every clique. The computed distances will differ from the exact ones only by an additive constant factor -- this leads to the approximation algorithm described in \cite{CGL24}. To devise an exact algorithm it would be sufficient to solve the following problem: given a geometric intersection graph on the set of points $V$, and given some subset $A \subseteq V$ of points lying very close to each other (e.g. $A$ fitting inside a square of small constant size $\delta$), compute and encode the neighbourhoods of all these vertices, in subquadratic time.

\smallskip 

\noindent \textbf{Unit-disk graphs.} Our data structure works with unit squares and general convex polygons, but leaves open a case of unit-disk intersection graphs:

\begin{problem} 
Is there a subquadratic algorithm for \textsc{Diameter} or $k$-\textsc{Diameter} for unit-disk intersection graphs?
\end{problem}

This time, the main obstacle is the Neighbouring Set Data Structure: our techniques does not seem to generalize to unit disks. We would need a new way of constructing such data structures.

\smallskip 

\textbf{Lower bounds.} Finally, we conjectured in Introduction that $\Ot{n^{1-1/d} \cdot m}$ is a candidate for a tight complexity bound. Let us generalize this question to any lower bounds for \textsc{Diameter}.
\begin{problem}
Are there any (conditional) lower bounds for \textsc{Diameter} and $k$-\textsc{Diameter} for either:
\begin{itemize}
    \item $K_t$-minor-free graphs;
    \item bounded distance VC-dimension graphs;
    \item \ldots or geometric intersection graphs?
\end{itemize}
\end{problem}

\bibliography{main}

\newpage
\appendix

\section{Full proofs and other technical details}
\label{appendix:omitted}
\subsection{Preliminaries}

 Let $(X,\Rr)$, $(X,\Rr')$ be two hypergraphs on the same underlying set $X$. Suppose that both of them have VC-dimension $d$. The following lemma is used in our paper to bound the dimension of the hypergraph $(X, \{R \sdiff R' \mid R \in \Rr, R' \in \Rr'\})$. A bound of $\Oh{d \log d}$ was proven in \cite{EA07}, but for the union operation $\cup$ instead of $\sdiff$ operation; nevertheless, almost the same proof works in the general case and the bound seems to be regarded as folklore. For the sake of completeness, we formulate and show the general result here, for any operator $\circ$ on set such that intersection distributes over $\circ$ (i.e. $A \cap (B \circ B') = (A \cap B) \circ (A \cap B')$ for any sets $A, B, B'$). The union, intersection, set difference and symmetric difference operators all have this property.

\lemmaJoinDim*
\begin{proof}
    Let $S \subset X$ and let $|S| = s$. From Corollary \ref{cor:ssp-subset} we know that $|\{R \cap S | R \in \Rr\}| \leq \beta s^d$ for some constant $\beta$ depending only on $d$. Similarly, $|\{R' \cap S | R' \in \Rr'\}| \leq \beta s^d$. For every $R^* \in \Rr^*$ we have $R^* = R \circ R'$ for some $R \in \Rr$, $R' \in \Rr'$, so by distributivity every subset $S \cap R^* = (S \cap R) \circ (S \cap R')$. Therefore $|\{R^* \cap S | R^* \in \Rr^*\}| \leq  (\beta s^d)^2$. Observe that we can pick a large enough constant $\gamma$ (in fact, $\gamma > 8$ should suffice) such that $\gamma \cdot d \log d > 2d \cdot (\log d + \log \log d + \log \gamma) + 2\beta$, so for $s = \gamma \cdot d \log d$ we have $2^s > \beta^2 \cdot s^{2d}$. Therefore, $\Rr^*$ cannot shatter a set larger than $\gamma \cdot d \log d$.
\end{proof}

\subsection{Geometric intersection graphs}

\lemmaCenterSymmetry*
\begin{proof}
    For any $v_1,v_2\in V$ shapes $v_1 + \mathcal{F}, v_2 + \mathcal{F}$ intersect if and only if
    \[(0,0) \in (\{v_2\} \oplus \mathcal{F}) \oplus -(\{v_1\} \oplus \mathcal{F}) \iff v_2 - v_1 \in \mathcal{F} \oplus (-\mathcal{F}) = 2 \mathcal{H}. \]
    Similarly, shapes $v_1 + \mathcal{H}, v_2 + \mathcal{H}$ intersect if and only if 
    \[v_2 - v_1 \in \mathcal{H} \oplus (-\mathcal{H}) = \mathcal{H} + \mathcal{H} = 2\mathcal{H},\]
    where the first equality follows from the fact that $\mathcal{H}$ is symmetric around $(0,0)$ (from definition) and the second one from the fact that $\mathcal{H}$ is convex. Therefore, both graphs are equivalent and conclusion follows.
\end{proof}

\lemmaBounded*

Before we prove this lemma, we need a preliminary result.  It turns out that $\mathcal{F}$ induces a norm and a metric on $\mathbb{R}^2$ in a natural manner: let $v_\mathcal{F}:\mathbb{R}^2\rightarrow [0,\infty)$ be defined as $v(x) = \min\{r\in \mathbb{R}_{\geq 0} \mid x\in r\cdot \mathcal{F}\}$. Because $\mathcal{F}$ is closed, this function is well defined and $x \in v_\mathcal{F}(x)\cdot \mathcal{F}$ for any point $x$.

\begin{lemma}
    The function $m_\mathcal{F}: \mathbb{R}^2 \to [0, \infty)$ given by $m_\mathcal{F}(x,y) = v_\mathcal{F}(x-y)$ is a metric on $\mathbb{R}^2$. Additionally, for any points $A, B, C \in\mathbb{R}^2$, if $B$ lies on segment $AC$, then we have $m_\mathcal{F}(A,C) = m_\mathcal{F}(A,B) + m_\mathcal{F}(B,C)$.
\end{lemma}

\begin{proof}
    For clarity, in this prove we use $m$ instead of $m_\mathcal{F}$. First, let us prove that $m_\mathcal{F}$ is a metric:
    \begin{itemize}
    \item $m(x,y) = 0 \iff x=y$ follows from the fact that $\mathcal{F}$ is bounded.
    \item Symmetry follows from the fact that $v(x) = v(-x)$ which is consequence of symmetry of $\mathcal{F}$ around $(0,0)$.
    \item We want $m(x,y) \leq m(x,z) + m(z,y)$. Without loss of generality we may assume $z = (0,0)$, so that it suffices to show $v(x-y) \leq v(x) + v(y)$. Let $r=v(x),s=v(y)$, so that $x\in r\mathcal{F}$ and $y\in s\mathcal{F}$. We need to prove $x-y \in (r+s)\mathcal{F}$. but $x = rf_x, y=sf_y$ for some $f_x,f_y \in \mathcal{F}$ and $(r+s)\mathcal{F} = r\mathcal{F} \oplus s\mathcal{F}$ by convexity. Because $-f_y\in \mathcal{F}$ by symmetry, we get $x - y = rf_x + s(-f_y) \in r\mathcal{F} + s\mathcal{F} = (r+s)\mathcal{F}$. 
    \end{itemize}
    
    For the second part, let $A,B,C \in \mathbb{R}^2$ with $B$ laying on $AC$. Denote by $|PQ|$ Euclidean length of the segment $PQ$ for any $P,Q\in\mathbb{R}^2$. Without loss of generality we can assume $A=(0,0)$. The boundary of $m(A,B)\mathcal{F}$ passes through $B$ and the boundary of $m(A,C)\mathcal{F}$ passes through $C$. By similarity of both shapes, we get: 
    \[\frac{m(A,B)}{m(A,C)} = \frac{|AB|}{|AC|}.\]
    Analogously:
    \[\frac{m(C,B)}{m(C,A)} = \frac{|CB|}{|CA|}.\]
    Therefore we can deduce that:
    \[\frac{m(A,B) + m(B,C)}{m(A,C)} = \frac{m(A,B)}{m(A,C)} + \frac{m(C,B)}{m(C,A)} = \frac{|AB|}{|AC|} + \frac{|CB|}{|CA|} = \frac{|AB| + |BC|}{|AC|} = 1,\]
    and the statement follows.
\end{proof}

We are now ready to show that the distance VC-dimension of intersection graphs is at most 4.

\begin{proof}[Proof of Lemma \ref{lemma:bounded}]
    By Lemma \ref{lemma:symmetric} we can assume that $\mathcal{F}$ is symmetric. Let $d(x,y)$ denote distance between points $x$ and $y$ with respect to the metric $m_\mathcal{F}$, while  $\ell(u,v)$ denotes the distance between vertices $u$ and $v$ in the intersection graph $G = I(V,\mathcal{F})$.
    
    By contradiction, let us assume that there is a set of 5 vertices $Y=\{v_1,v_2,v_3,v_4,v_5\}$ of $V$ which is shattered by the ball hypergraph.
    
    In particular, for every pair of indices $i,j$ such that $1\leq i < j\leq 5$ there is a vertex $x_{ij}$ and an integer $k_{ij}$ such that $N^{k_{ij}}[x_{ij}]\cap Y = \{v_i,v_j\}$.
    For any pair $(i,j)$ let as choose arbitrary shortest-length paths  $x_{ij} \to v_i$ and $x_{ij} \to v_j$ in the intersection graph, and connect their consecutive points with segments (see Figure \ref{fig:squares1}). After doing so for every pair of indices, we are left with a drawing of $K_5$ on plane.
    
    Because of non-planarity of $K_5$, some two paths on this drawing intersect. We can assume without loss of generality that the intersecting paths are $v_1 \to x_{12} \to v_2$ and $v_3 \to x_{34} \to v_4$. Furthermore, we can assume that the specific parts $v_1 \to x_{12}$ and $v_3 \to x_{34}$ intersect. Let $a = x_{12}$, $u_a = v_1$, $b = x_{34}$, $u_b = v_3$ and let $k_a = k_{12} = \ell(a,u_a)$ and $k_b = k_{34} = \ell(b,u_b)$. We know that the paths $a \to u_a$ and $b \to u_b$ cross, but there are two cases: either there is a common vertex $q$ lying on both paths, or there is no such vertex.

    First, we prove that $q$ cannot exist. If it does, we can assume without loss of generality that $\ell(q,u_a) \geq \ell(q,u_b)$. But this implies $\ell(a,u_a) = \ell(a,q) + \ell(q,u_a) \geq \ell(a,q) + \ell(q, u_b) \geq \ell(a, u_b)$. Therefore $k_{12} = d(x_{12}, v_1) \geq \ell(x_{12}, v_3)$, which contradicts $N^{k_{12}}[x_{12}]\cap Y = \{v_1,v_2\}$.
    
    If the paths do not have a common vertex, then there is a point $p$ which is a crossing between some segments $(a',a'')$ and $(b',b'')$ belonging to paths $a \to u_a$ and $b \to u_b$ respectively (see Figure \ref{fig:squares2}). Of course, $d(a',a'') \leq 1$ and $d(b',b'') \leq 1$. From the triangle inequality and the fact that $p$ lies on both segments $a'a''$ and $b'b''$ we deduce:
    \[2\geq d(a',a'') + d(b',b'') = d(a',p) + d(p,a'') + d(b',p) + d(p,b'') \geq d(a',b'') + d(b',a''),\]
    
    We conclude that $d(a',b'') \leq 1$ or $d(b',a'')\leq 1$. Analogously, $d(a',b')\leq 1$ or $d(a'',b'')\leq 1$.

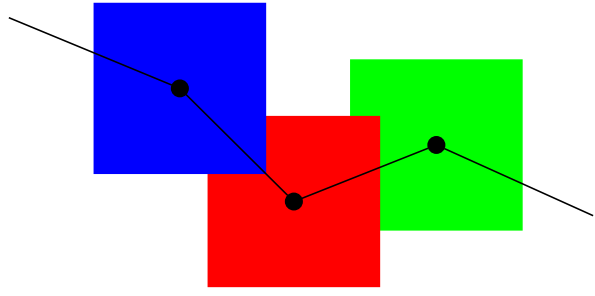
\begin{figure}
\begin{center}
    
\scalebox{1.5}{\begin{tikzpicture}
	\begin{pgfonlayer}{nodelayer}
		\node [style=none] (0) at (-8, 7) {};
		\node [style=none] (1) at (-5, 7) {};
		\node [style=none] (2) at (-8, 4) {};
		\node [style=none] (3) at (-5, 4) {};
		\node [style=none] (4) at (-6, 2) {};
		\node [style=none] (5) at (-6, 5) {};
		\node [style=none] (6) at (-3, 5) {};
		\node [style=none] (7) at (-3, 2) {};
		\node [style=new style 0] (8) at (-4.5, 3.5) {};
		\node [style=new style 0] (9) at (-6.5, 5.5) {};
		\node [style=none] (10) at (-3.5, 6) {};
		\node [style=none] (11) at (-3.5, 3) {};
		\node [style=none] (12) at (-0.5, 6) {};
		\node [style=none] (13) at (-0.5, 3) {};
		\node [style=new style 0] (14) at (-2, 4.5) {};
		\node [style=none] (15) at (-9.5, 6.75) {};
		\node [style=none] (16) at (0.75, 3.25) {};
	\end{pgfonlayer}
	\begin{pgfonlayer}{edgelayer}
		\draw [style=new edge style 4] (11.center)
			 to (13.center)
			 to (12.center)
			 to (10.center)
			 to cycle;
		\draw [style=new edge style 5] (7.center)
			 to (4.center)
			 to (5.center)
			 to (6.center)
			 to cycle;
		\draw [style=new edge style 6] (3.center)
			 to (1.center)
			 to (0.center)
			 to (2.center)
			 to cycle;
		\draw (9) to (8);
		\draw (8) to (14);
		\draw (14) to (16.center);
		\draw (9) to (15.center);
	\end{pgfonlayer}
\end{tikzpicture}}
\end{center}
\caption{Connecting centers of squares on a path.}
\label{fig:squares1}
\end{figure} 

    Without loss of the generality we assume that $d(a',b'')\leq 1$. Therefore $a'$ and $b''$ are connected by an edge in $G$. We have the following inequalities:
    \begin{itemize}
        \item $\ell(a,a') + 1+  \ell(a'',u_a) \leq k_a,$
        \item $\ell(b,b') + 1 + \ell(b'',u_b) \leq k_b.$
    \end{itemize}

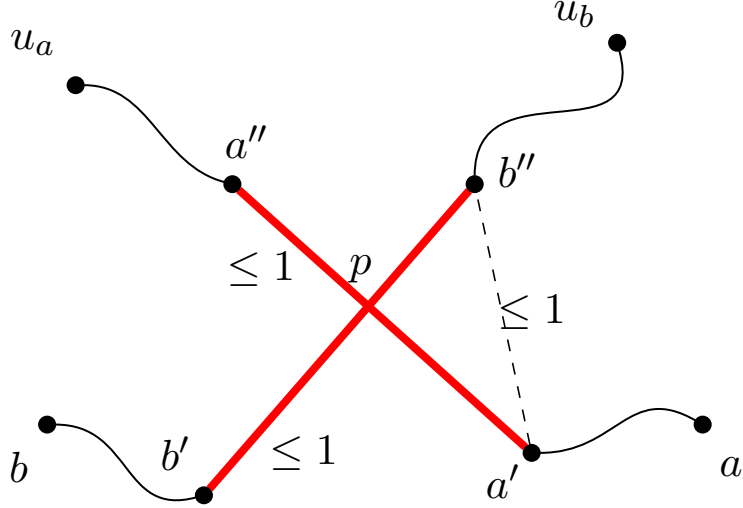
\begin{figure}

\begin{center}
    
\scalebox{1.5}{\begin{tikzpicture}
	\begin{pgfonlayer}{nodelayer}
		\node [style=new style 0] (0) at (-3.5, 4) {};
		\node [style=new style 0] (1) at (1.75, -0.75) {};
		\node [style=new style 0] (2) at (0.75, 4) {};
		\node [style=new style 0] (3) at (-4, -1.5) {};
		\node [style=none] (4) at (-1.25, 2.5) {$p$};
		\node [style=none] (5) at (1.25, -1.25) {$a'$};
		\node [style=none] (6) at (-3.25, 4.75) {$a''$};
		\node [style=none] (7) at (1.5, 4.25) {$b''$};
		\node [style=none] (8) at (-4.5, -0.75) {$b'$};
		\node [style=none] (9) at (-2.25, -0.75) {$\leq 1$};
		\node [style=none] (10) at (-3, 2.5) {$\leq 1$};
		\node [style=none] (11) at (1.75, 1.75) {$\leq 1$};
		\node [style=new style 0] (12) at (4.75, -0.25) {};
		\node [style=new style 0] (13) at (3.25, 6.5) {};
		\node [style=new style 0] (14) at (-6.25, 5.75) {};
		\node [style=new style 0] (15) at (-6.75, -0.25) {};
		\node [style=none] (16) at (5.25, -1) {$a$};
		\node [style=none] (17) at (-7, 6.5) {$u_a$};
		\node [style=none] (18) at (-7.25, -1) {$b$};
		\node [style=none] (19) at (2.5, 7) {$u_b$};
	\end{pgfonlayer}
	\begin{pgfonlayer}{edgelayer}
		\draw [style=new edge style 7] (1) to (2);
		\draw [style=new edge style 8, in=0, out=150, looseness=1.25] (12) to (1);
		\draw [style=new edge style 8, in=-75, out=90, looseness=1.50] (2) to (13);
		\draw [style=new edge style 8, in=0, out=165] (0) to (14);
		\draw [style=new edge style 8, in=0, out=-165, looseness=1.25] (3) to (15);
		\draw [style=new edge style 9] (3) to (2);
		\draw [style=new edge style 9] (1) to (0);
	\end{pgfonlayer}
\end{tikzpicture}}
\end{center}
    \caption{Situation in the proof of bound on distance VC-dimension. Edges between vertices $a',a''$ and $b',b''$ imply that there is an edge between vertices $a',b''$ or $b',a''$.}
\label{fig:squares2}
\end{figure}

    Now, if $\ell(a,a') + 1 + \ell(b'',u_b) \leq k_a$ then the existence of the path $a\rightarrow a' \sim b'' \rightarrow u_b$ of length at most $k_a$ proves that set $N^{k_a}(a)$ contains at least 3 vertices from $Y$, contradiction.
    On the other hand, if $\ell(a,a') + 1 + \ell(b'',u_b) > k_a$, this means  $\ell(a,a') + \ell(b'',u_b)\geq k_a$, hence:
    \[\ell(b,u_a) \leq \ell(b,b') + \ell(b',a'') + \ell(a'',u_a) \leq k_b - 1 - \ell(b'',u_b) + 2 + k_a - 1 - \ell(a,a') \leq k_b,\]
    where the inequality $\ell(b',a'')\leq 2$ follows from the fact that either $(a',b')$ or $(a'',b'')$ is an edge in $G$, so we there is a path $b'\sim a'\sim a''$ or $b'\sim b''\sim a''$.
    Finally, we conclude $u_a\in N^{k_b}[b]$ and by contradiction, the statement follows.
\end{proof}

\subsection{General algorithm framework}


\corDistanceHypergraphXors*
\begin{proof}
Directly from Lemma \ref{lemma:hypergraph-order} with $X = V$ and $\Rr = \{N^k[v] : v \in V\}$. Note that for the algorithmic part we use the fact that for any $v \in V$ and $k \in \mathbb{Z}$,  computing the list $\{x \in V : v \in N_k[x] \}$ of hyperedges containing $v$ is equivalent to computing $N_k[v]$, and so can be done in time complexity $T(G)$.
\end{proof}

\corWeightedAverageSpanningPath*
\begin{proof}
We build a new hypergraph $(X, \Rr)$ in the following way: for every vertex $v \in V$ we add $\alpha(v)$ copies of $v$, named $v^{(1)}, \ldots, v^{(\alpha(v))}$ to $X$. This implies $|X| = \sum_{j=1}^{n} \alpha(j)$. Then for every $v \in V$ we create a hyperedge $e(v)$ of $\Rr$ defined as $e(v) = \{x^{(t)} : x \in N^k[v], 1 \leq t \leq \alpha(x)\}$. In other words, $(X, \Rr)$ is the $k$-neighbourhood hypergraph of $V$, but with every vertex multiplied according to its weight. However, multiplication of vertices cannot increase the VC-dimension, as it does not create any new shattered sets. Therefore the VC-dimension of $(X, \Rr)$ is at most $d$ and we can apply Lemma \ref{lemma:hypergraph-order} to it, obtaining some order $e_1, \ldots, e_n$ on $\Rr$. We can now assign names $v_1, v_2, \ldots, v_n$ such that $e_i = e(v_i)$, and we are going to prove that $v_1, \ldots, v_n$ is the desired order on $V$. Computing this order, like in Corollary \ref{corollary:distance-hypergraph-xors} needs $\Oh{n^{1/d}}$ applications of single-source distance finding algorithm. Note that we don't need to explicitly copy vertices, but simply sample original vertices with probability proportional to $\alpha$.

From Lemma \ref{lemma:hypergraph-order} we know that $\sum_{i=1}^{n-1}|e_i \sdiff e_{i+1}| = \Ot{n^{1-1/d} \cdot \sum_{j=1}^{n} \alpha(j)}$. Recall that for every $v \in V$ we denote by $I(N^k[v])$ the interval representation of $N^k[v]$. For a fixed $x$, let us consider right endpoints of intervals in $I(x)$, excluding (if needed) the endpoint at the very last element $v_n$.
Any such endpoint contributes to a situation where some vertex $v_i$ belongs to $N^k[x]$, but $v_{i+1}$ does not, which is equivalent to $x \in N^k[v_i]$ and $x \notin N^k[v_{i+1}]$, which in turn means that $\{x^{(1)}, \ldots, x^{(\alpha(x))}\} \subseteq e(v_i) \setminus e(v_{i+1})$. In a similar manner we can show that any left endpoint $v_j$ of an interval in $I(x)$ contributes $\alpha(x)$ elements to $e(v_j) \setminus e(v_{j-1})$. Summing it over all $x$, by the fact that any interval has one left and one right endpoint, we get that: 
\[\sum_{x \in V} \alpha(x) \cdot |I(N^k[x])| = \sum_{i=1}^{n-1}|e_i \sdiff e_{i+1}| + \ell + r\] 
where $\ell$ and $r$ count the endpoints at the leftmost and rightmost spot, i.e. $\ell$ (resp. $r$) is the sum of $\alpha(y)$ for all $y$ such that $I(y)$ has an endpoint at $v_1$ (resp. $v_n$). But as $\ell$ and $r$ are both bounded by $\sum_{j=1}^{n} \alpha(j)$, so we finally derive: \[\sum_{x \in V} \alpha(x) \cdot |I(N^k[x])| = \Ot{n^{1-1/d} \cdot \sum_{j=1}^{n} \alpha(j)},\]

as desired.

\end{proof}

\subsection{Diameter testing for implicit graphs}

\lemmaExpandBalls*
\begin{proof}
	We first consider the time consumed by the procedure before recursive calls.
	The computation of sets $C = \bigcup_{i=1}^t D_i$, $D_1 \setminus C$, $D_1'$ and $D_m'$ can be easily implemented in $\Ot{t + \sum_{i=1}^{t} |D_i|}$ time.
	There are $|D_1 \setminus C|$ $\AddNeighbours$ operations, each taking $\Ot{1}$ time.
	Therefore the total time spent before recursion is $\Ot{t + \sum_{i=1}^{t} |D_i|}$.
	To simplify further analysis, we define cost of these operations to be $t + \sum_{i=1}^{t} |D_i|$, i.e. we omit the polylogarithmic factors introduced by $\widetilde{\mathcal{O}}$ notation.
	Our goal is to bound the total cost including recursive calls.
	The total running time is then simply the cost multiplied by polylogarithmic factors.

	Let $f(a, b, t)$ be the maximum possible total cost of the recursive procedure if $|D_1| = a$ and $\sum_{i=2}^{t} |D_i| = b$.
	Observe that $D_1', D_m' \subseteq C$, so size of these sets is bounded by $b$.
	Let $x = \sum_{i=2}^{m-1} |D_i|$.
	Then the cost of first recursive call is at most:
	\begin{equation*}
		f(|D_1'|, \sum_{i=2}^{m-1} |D_i|, \floor{t/2}) \leq f(b, x, \floor{t/2})
	\end{equation*}
	and the cost of second recursive call is at most:
	\begin{equation*}
		f(|D_m'|, \sum_{i=m+1}^{t} |D_i|, \floor{t/2}) \leq f(b, b-x, \ceil{t/2}) \text{.}
	\end{equation*}
	For $t \geq 2$, the following relation follows:
	\begin{equation*}
		f(a, b, t) \leq \max_{x=0}^{b} f(b, x, \floor{t/2}) + f(b, b-x, \ceil{t/2}) + a + b + t
	\end{equation*}
	We argue that $f(a, b, t) \leq a + 3b(\ceil{\log_2 t} + 1) + 2t$.
	The base case $f(a, b, 1)$ is trivial.
	We now prove the inductive step for $t \geq 2$:
	\begin{align*}
		f(a, b, t) &\leq \max_{x=0}^{b} f(b, x, \floor{t/2}) + f(b, b-x, \ceil{t/2}) + a + b + t \\
		&\leq \max_{x=0}^{b} b + 3x\ceil{\log_2 t} + 2\floor{t/2} + b + 3(b-x)\ceil{\log_2 t} + 2\ceil{t/2} + a + b + t \\
		&= 3b + 3b\ceil{\log_2 t} + a + 2t = a + 3b(\ceil{\log_2 t} + 1) + 2t
	\end{align*}
	It follows that running time is $\Ot{f(a, b, t)} = \Ot{t + \sum_{i=1}^{t} |D_i|}$.
\end{proof}

\subsection{Simulating BFS using Neighbour Set Data Structure}
\label{sec:explicit-ball-computation-using-ds}

The vertex order computation requires an oracle that explicitly computes the set $N^r[v]$ for a given radius $r$ and vertex $v$.
For completeness, we describe here how to do that in $\Ot{n}$ time if we are equipped with Neighbour Set Data Structure.
More precisely, we show how to simulate Breadth-First Search algorithm on input graph without dependence on edge count.
This enables us to compute distances from vertex $v$ to all other vertices.

We provide the pseudocode of the procedure as Algorithm $\ref{alg:balls-explicit}$.
Classic BFS algorithm maintains a queue of vertices to visit $Q$ and a set of explored vertices $\widetilde{S}$.
The only important change we make is keeping the set $\widetilde{S}$ in a data structure.
This allows us to quickly list only unexplored neighbours for any vertex $x \in V$ in line \ref{line:listing-neighbours}.
The listed vertices are immediately marked as explored.
This guarantees that each vertex will be listed exactly once and the time complexity becomes $\Ot{n}$.

\begin{algorithm}[H]
	\caption{Simulating Breadth-First Search using data structure.}
	\label{alg:balls-explicit}
	\begin{algorithmic}[1]
		\Function{GetNeighbourhood}{$v, r$}
			\State $\widetilde{S} \gets \emptyset$ \Comment{The set of unexplored vertices.}
			\State $Q \gets \text{empty FIFO queue}$ \Comment{The queue of vertices to visit.}
			\ForAll{$x \in V(G) \setminus \{v\}$}
				\State $d_x \gets \infty$ \Comment{Initialize distances to $\infty$.}
			\EndFor
			\State $d_v \gets 0$ \Comment{Source vertex is at distance $0$.}
			\State Push vertex $v$ to the queue $Q$.
			\While{$Q$ is not empty}
				\State $x \gets \text{pop vertex from } Q$
				\State $\widetilde{S'} \gets \AddNeighbours(\widetilde{S}, x)$
				\ForAll{$w \in \ListDifferences(\widetilde{S}, \widetilde{S'}) \setminus \{v\}$}\label{line:listing-neighbours}
					\State $d_w \gets d_x + 1$
					\State Push vertex $w$ to the queue $Q$.
				\EndFor
				\State $\widetilde{S} \gets \widetilde{S'}$
			\EndWhile
			\State \Return $\{ x \in V(G) : d_x \leq r \}$ \Comment{Return $r$-neighbourhood.}
		\EndFunction
	\end{algorithmic}
\end{algorithm}

\section{Neighbouring Set Data Structure for intersection graphs}
\label{appendix:squares}

In this section we provide more details for the Neighbouring Set Data Structure needed for our diameter algorithms. 

Before we introduce the data structures, let us recall the following lemma and complete its omitted proof:

\obsDoubling*

\begin{proof}
    By symmetry of $\mathcal{F}$ we get 
     \[(u + \mathcal{F}) \cap (v + \mathcal{F}) \neq \varnothing \iff u-v \in \mathcal{F} \oplus (-\mathcal{F}) = \mathcal{F} \oplus \mathcal{F} = 2\mathcal{F} \iff u \in v + 2\mathcal{F},\]
 which proves the claim.
\end{proof}

Recall that we are given a shape $\mathcal{F} \subset \mathbb{R^2}$ and a set of points $V \subset \mathbb{R}^2$ and we want the NSDS to store a family $\mathcal{T}$ of subsets of $V$. It should implement the following operations:

\begin{itemize}
    \item $\Mark(\widetilde{S}, (x, y))$: Given a set $\widetilde{S} \in \mathcal{T}$ and a point $(x,y) \in \mathbb{R}^2$, add a new set $\widetilde{S'} = \widetilde{S} \cup P$ to the family $\mathcal{T}$, where $P \subseteq V$ contains the points covered by $\mathcal{F}$ centered at the point $(x,y)$.
    \item $\ListDifferences(\widetilde{S}_1, \widetilde{S}_2)$: Given sets $\widetilde{S}_1, \widetilde{S}_2 \in \mathcal{T}$, output their symmetric difference $D = \widetilde{S}_1 \sdiff \widetilde{S}_2$.
\end{itemize}

In the following, we assume for simplicity that there are no ties when coordinates are compared during the algorithm. We can achieve that e.g. by rotating the plane by a random angle, but it is not necessary: we can keep the original points and add a few edge cases to the following analysis. We will, however, omit them here for clarity. Let us restate the main lemma which needs to be proven here:

\lemmaStructureGeo*

\subsection{Segment trees}
\label{sec:persistent-segment-tree}

\emph{Segment trees} are a family of data structures that facilitate efficient range queries and updates over a sequence of elements. A segment tree is typically build over some sequence of objects $x_1,x_2,\ldots,x_n$ in a recursive manner. If the sequence consists of single object, then segment tree consists of a single \textit{node}. Otherwise, we split the sequence into $x_1,\ldots,x_m$ and $x_{m+1}, \ldots, x_n$ and build segment trees for these subsequences separately. We choose the index $m$ such that sequence is split into equal or nearly-equal parts. (In fact, it is often assumed that $n$ is a power of $2$ and that $m = n/2$) .After that, we create a new node, and declare nodes created by recursive the construction as \textit{left} and \textit{right child} of new node, respectively.

From this definition it follows that every node has 0 or 2 children. If a node has no children, then it is called a \textit{leaf}. If a node $x$ is a child of node $y$ then we say that $y$ is the \textit{parent} of node $x$. There is exactly one node with no parent, called  \textit{root} of the tree. Moreover, we can observe that tree has $\Oh{\log{n}}$ levels where $n$ is length of initial sequence.

To facilitate its operations, a segment tree stores information in its nodes. In particular, if a node is responsible for some interval of the sequence $\{x_a, x_{a+1}, \ldots, x_b\}$, then it stores some cumulative information about objects in that interval.

For our purposes, we also need \textit{persistency}, i.e. we want to have access to past versions (snapshots) of the tree. It can be achieved by copy-on-write mechanism: if we are to modify a node, we instead copy it, and modify this copy. As our implementation of operations on segment tree typically starts with root and then goes down the tree, every modification will create a new root of the tree, which can be later used to see state of the tree when the root was created.

As an example application we now describe a data structure called \emph{Simple Subset Retrieval} (SSR) maintaining a family $\mathcal{T}$ of subsets of $\{1,2,\ldots,n\}$ and supporting following operations:

\begin{itemize}
    \item $\Initialize(n)$: Start with $\mathcal{T} = \{\varnothing\}$ and return the representative of $\varnothing$, in $\Ot{n}$ time complexity;
    \item $\Add(\widetilde{S}, k)$: For a set $\widetilde{S} \in \mathcal{T}$, add a new set $\widetilde{S}' = \widetilde{S} \cup \{k\}$ to $\mathcal{T}$ and return the representative of $\widetilde{S}'$ in $\Ot{1}$ time complexity;
    \item $\ListDifferences(\widetilde{S_1}, \widetilde{S_2})$: Return the symmetric difference $D = \widetilde{S_1} \sdiff \widetilde{S_2}$ in $\Ot{|D|}$ time complexity.
\end{itemize}

The full version of the Neighbour Set Data Structure will use similar ideas -- therefore the SSR will build helpful intuitions -- but also a version of the SSR will be used as part of NSDS.

The Simple Set Retrieval structure consists of a single persistent segment tree. The representatives of sets in $\mathcal{T}$ are the new roots created by the aforementioned persistency mechanism.

To implement $\ListDifferences$ operation we need a fast way of determining if two sets have the same elements. To achieve this, we use hashing.
For every $i \in \{1,\ldots,n\}$ we randomly select a $t$-bit hash $h_i$. The hash of the subset $S$ of $\{1,\ldots,n\}$ is $\bigoplus_{s\in S} h_s$ where $\oplus$ is the bitwise exclusive or (XOR) operation (we assume the result to be 0 if $S$ is empty).

Each node is responsible for elements from some contiguous subset $\{a,a+1,\ldots,b\}$ of $\{1,\ldots,n\}$ and stores only a single field \texttt{hash}. We would like \texttt{hash} to hold a hash of set $\widetilde{S} \cap \{a,\ldots,b\}$. Therefore the value of \texttt{hash} in the leaf responsible for $\{i\}$ is either $h_i$ if the element $i$ is already added to the set or 0 otherwise. If the node is not a leaf, then the \texttt{hash} field can be calculated as xor of \texttt{hash} fields from both of its children. 
As children partition set $\{a,\ldots,b\}$ this calculation is indeed correct.

For node $N$ responsible for points $\{a,\ldots,b\}$ define $S(N)$ as set of added elements amongst $\{a,\ldots,b\}$.

The operations can be performed as follows:
\begin{itemize}
    \item $\Initialize()$. We randomly select $t$-bit numbers $h_i$. Then we build the tree recursively. Let $\InitTree(a,b)$ be a procedure which creates the subtree responsible for interval $\{a,\ldots,b\}$ and returns the root of the created subtree. The procedure works as follows. First, create a new empty node. If $a$ is equal to $b$, then new node is a leaf and we can simply return it. Otherwise we split interval $[a,b]$ in half with $m = \floor{(a+b)/2}$ and recursively create node for intervals $[a,m]$ and $[m+1,b]$. The nodes returned by these recursive calls are now roots of the left and right subtrees. The set is initially empty, therefore all \texttt{hash} values are set to 0. The correctness of this procedure follows from a simple induction on the length of the interval. The node returned by $\InitTree(1,n)$ is the root of the tree which represent the empty set.
    \item $\Add(\widetilde{S}, k)$. The representative of $\widetilde{S}$ is simply one of the copies of the root. Our goal is to change value in the leaf responsible for $\{k\}$ while maintaining persistency. We can do so recursively as follows. If a node $N$ has children, then we can determine the child $C$ which is responsible for the element $k$. We proceed recursively on the child $C$. In return we get a new node $C'$: copy of the original child $C$ but now its hash includes the element $k$. We create a copy $N'$ of $N$. We link the node $C'$ returned by the recursive call as the appropriate child of $N'$ and recalculate the \texttt{hash} value of $N'$ using values stored in its children (one of the children is new, the other one is from the original node $N$).
    \item $\ListDifferences(\widetilde{S_1}, \widetilde{S_2})$.
    Now the representatives $\widetilde{S_1}$ and $\widetilde{S_2}$ are the roots of two snapshots of our data structure. We again proceed recursively. Suppose we now consider nodes $N_1$ and $N_2$. We want to return symmetric difference $S(N_1)\Delta S(N_2)$.
    If the hashes in the nodes $N_1$ and $N_2$ are equal, then (w.h.p.) the sets $S(N_1)$ and $S(N_2)$ are the same. In this case we can return $\varnothing$ as our result. Suppose now that hashes in $N_1$ and $N_2$ are different. If the nodes $N_1$ and $N_2$ are leaves, then they are responsible for a single element $x$ and $\widetilde{S_1}$ and $\widetilde{S_1}$ differ on $x$, so we should return $\{x\}$. Otherwise, we simply use recursive call on both children and return a union of the sets returned from the both recursive calls. 
\end{itemize}

The proof of complexity of above procedures is very similar to one given in next part, therefore it is omitted here.

\subsection{Single Stripe Data Structure}
\label{sec:stripes-structure}

Here we introduce an auxiliary data structure which behaves similarly to NSDS, but assumes that all points lie in a single horizontal stripe of plane $[0;1]\times \mathbb{R}$. We will call it the Single Stripe Data Structure (SSDS).
In this and the following subsections, we work on Lemma \ref{lemma:structure-geo}a), so we assume that our shape $\mathcal{F}$ is a unit square. In Subsection \ref{sec:full-geo-structure} we generalize our approach to any $s$-sided convex polygon.

Let $V$ be a set of $n$ points from a stripe $[0;1]\times \mathbb{R}$. The structure maintains a family $\mathcal{T}$ of subsets of $V$. The structure supports following operations.

\begin{itemize}
    \item $\Initialize(V)$: Adds the empty set $\varnothing$ to the family $\mathcal{T}$ and return its representative in $\Ot{n}$  time;
    \item $\Mark(\widetilde{S},(x,y))$: Let $P$ be a set of points from $V$ contained within a unit square centered at the point $(x,y)$. This operation adds the set $\widetilde{S}' = \widetilde{S} \cup P$ to the family $\mathcal{T}$ and return the representative of the set $\widetilde{S}'$ in $\Ot{1}$ time;
    \item $\ListDifferences(\widetilde{S_1},\widetilde{S_2})$: Returns the symmetric difference $D = \widetilde{S_1} \sdiff \widetilde{S_2}$ in $\Ot{|D|}$ time.
\end{itemize}

We start by sorting the set $V$ by the $x$ coordinate, so we assume that $V = v_1,v_2,\ldots,v_n$ follows this order. This way, points between any two vertical lines correspond to an interval in $v_1,\ldots,v_n$.

The SSDS data structure is a single persistent segment tree. Each node is responsible for points $v_a,v_{a+1},\ldots v_b$ for some $a,b$ and each leaf is responsible for a single point from the set $V$. As before, we use copies of the root as representatives of sets.

During $\Initialize()$, we also randomly select numbers $h_i$ consisting of $t$ bits for each point $v_i$. Let $S_I = \{v_i \mid i \in I\}$. Then hash of the set $S_I$ is defined as $H(S_I) = \bigoplus_{i\in I} h_i$.

Observe that all relevant unit squares from $\Mark$ queries have to intersect either the bottom or the top side of the stripe. We call a unit square \textit{bottom} if it intersects the bottom side of the stripe. Otherwise we call such square \textit{top}. In particular, if the square intersects both sides, we assume that it is a bottom one. The \textit{top boundary} is the bottom border of the shape consisting of a union of all top squares and top side of the stripe. Analogously we define the \textit{bottom boundary}. We can see that the marked points are exactly those above the top boundary or below the bottom boundary (see Figure \ref{fig:borders}).

\begin{figure}
    \begin{center}
        \includegraphics[scale=1.3]{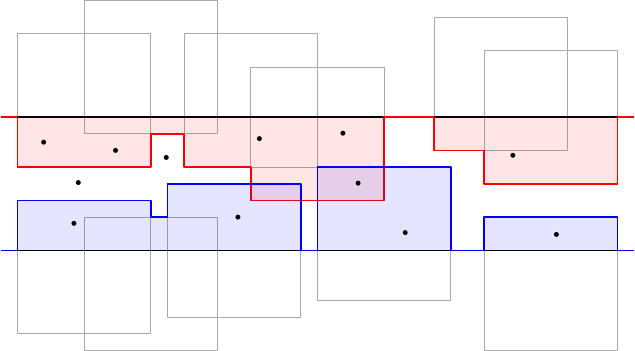}
    \end{center}
    \caption{Boundaries implied by squares overlapping with the stripe. The red and blue lines mark top and bottom boundaries respectively. The boundaries can intersect.}
    \label{fig:borders}
\end{figure}

Our goal is to calculate, for each node, a hash of the set of marked points for which this node is responsible for. The operation $\Mark$ can potentially affect many (even $\Theta(n)$) points from $V$, therefore we cannot afford to update every leaf directly. To address this problem we employ the technique of \textit{lazy propagation}. Namely, we do not have to keep accurate information in every node. Instead, we need to have correct information in a node only when we are accessing it.
We can update nodes in a lazy manner, i.e. mark a node if its children can potentially have inaccurate information. If needed, before entering any of the children, we update their values to correct state.

To facilitate the updates, in every node we store the following information (see Figure \ref{fig:values}):

\begin{enumerate}
    \item \texttt{points}: the array storing points $v_a,\ldots,v_b$, sorted by $y$ coordinate; let $j_0=0$ and $j_1,j_2,\ldots,j_{b-a+1}$ be indices $i$ of $v_i$ in this order; note that this array is never modified, so we do not need make any copies of it, instead keeping a single static version;
    \item \texttt{pref\_hash}: the array of hashes of the prefixes of \texttt{points}, i.e. \texttt{pref\_hash[0]} $=0$ and \texttt{pref\_hash[m]} $= \bigoplus_{i=1}^m h_{j_i}$; again, we only need one such array for all copies of the node;
    \item \texttt{top\_min}: the lowest $y$ coordinate on top boundary for $x\in [v_a^{(x)},v_b^{(x)}]$;
    \item \texttt{top\_max}: the highest $y$ coordinate on top boundary for $x\in [v_a^{(x)},v_b^{(x)}]$;
    \item \texttt{bot\_min}: the lowest $y$ coordinate on bottom boundary for $x\in [v_a^{(x)},v_b^{(x)}]$;
    \item \texttt{bot\_max}: the highest $y$ coordinate on bottom boundary for $x\in [v_a^{(x)},v_b^{(x)}]$;
    \item \texttt{top\_hash}: the hash of set of points from $v_a,\ldots,v_b$ above top boundary;
    \item \texttt{bot\_hash}: the hash of set of points from $v_a,\ldots,v_b$ below bottom boundary;
    \item \texttt{hash}: the hash of the set of all marked points from $v_a,\ldots,v_b$;
    \item \texttt{top\_lazy}, \texttt{bot\_lazy}: boolean flags used for \textit{lazy propagation}.
\end{enumerate}

\begin{figure}
    \begin{center}
        \includegraphics[width=\textwidth]{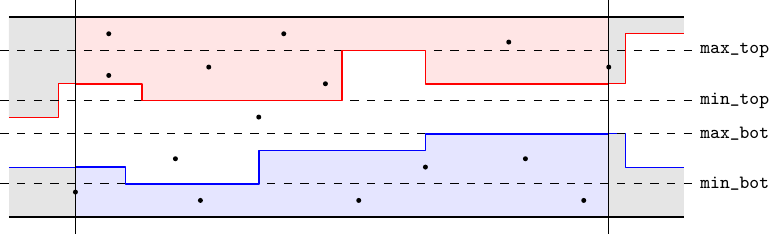}
    \end{center}
    \caption{Information about the boundaries stored by a segment tree's node. The \texttt{top\_hash} is the hash of the points in the red area. The \texttt{bot\_hash} is the hash of the points in the blue area. The \texttt{hash} is the hash of the points in the union of red and blue area.}
    \label{fig:values}
\end{figure}
\noindent
We need some nomenclature regarding types of nodes. We call a node $N$:
\begin{itemize}
    \item \textit{bottom-} or \textit{top-lazy} if flag \texttt{bot\_lazy} or \texttt{top\_lazy} is set to true, respectively;
    \item \textit{lazy} if it is either bottom-lazy or top-lazy;
    \item \textit{bottom-} or \textit{top-simple} if \texttt{bot\_min} = \texttt{bot\_max} or \texttt{top\_min} = \texttt{top\_max} respectively (the bottom or the top boundary is a horizontal line at this segment);
    \item \textit{$y$-disjoint} for some $y\in\mathbb{R}$ if the horizontal line at level $y$ is disjoint from both top and bottom boundary;
    \item \textit{split} if top and bottom boundaries are disjoint;
    \item \textit{bottom-outdated} (\textit{top-outdated}) if there is bottom-lazy (top-lazy) node on path from root to node $N$, excluding $N$;
    \item \textit{outdated} if it is either bottom-outdated or top-outdated.
\end{itemize}

\noindent
We are going to keep following invariants:

\begin{enumerate}
    \item After a node is created, information stored inside it never changes (persistency),
    \item If a node has an incorrect value of \texttt{bot\_min}, \texttt{bot\_max} or \texttt{bot\_hash}, then it must be bottom-outdated,
    \item If a node has an incorrect value of \texttt{top\_min}, \texttt{top\_max} or \texttt{top\_hash}, then it must be top-outdated,
    \item If a node has in incorrect value of \texttt{hash}, then it must be outdated,
    \item If we enter a node (using either operation), then it is not outdated,
    \item If a node is bottom-lazy (top-lazy) then it is bottom-simple (top-simple) and split.
\end{enumerate}

Observe that we can use \texttt{pref\_hash} array to calculate hash of any subarray of \texttt{points}. Let $I_{l,r} = \{j_l,j_{l+1},\ldots,j_r\}$. Then $H(I_{l,r}) = H(I_{0,r}) \oplus H(I_{0,l-1})$, with both these values stored in \texttt{pref\_hash} array. We use this property to recalculate hashes for nodes that become bottom- or top-simple during updates (possibly lazy updates).

We are now going to describe operations supported by SSDS. We introduce an auxiliary procedure $\Push$ which is part of the lazy propagation technique. Its purpose is to make sure that a node is not outdated before we use it.

\subparagraph*{\Push{} operation.} We use this operation a lazy node $N$ with potentially outdated children. The procedure returns an equivalent node $N'$ that is not lazy and its children are not outdated. As everything is persistent, this operation creates copies of the children, updates them, and returns a copy of the original node with new children.

If $N$ is top-lazy, then by invariants 5 and 6 it is not outdated and additionally it is top-simple and split. Therefore, correct states of its children are also top-simple and split. We can copy values \texttt{top\_min}, \texttt{top\_max} from $N$ to its children and mark them as top-lazy. 
Moreover, we can compute the \texttt{top\_hash} fields for the children of $N$ using their \texttt{pref\_hash} arrays. Observe that after this operation every node in the subtree of $N$, excluding its children, is still outdated (as it should be). Additionally, the children are no longer top-outdated.

Now we recalculate the \texttt{hash} fields for children. There are two cases. If the top boundary is below the bottom boundary, then every point is marked, so  we can set the \texttt{hash} field to the hash of all points. Otherwise, the points marked by the bottom boundary are different from the ones marked by the top boundary, in which case \texttt{hash} $=$ \texttt{bot\_hash} $\oplus$ \texttt{top\_hash}.

We proceed similarly if the node $N$ is bottom-lazy.
It is straightforward to check that all invariants are maintained by this operation.


\begin{algorithm}
    \caption{$\Push$ operation}
    \begin{algorithmic}
        \Function{Push}{$N$}
        \State $N \gets$ copy of $N$ 
        \If{$N.\texttt{top\_lazy}$}
        \State $N \gets \PushTop(N, \texttt{top\_min})$
        \State $N.\texttt{top\_lazy} \gets \texttt{false}$
        \EndIf
        \If{$N.\texttt{bot\_lazy}$}
        \State $N \gets \PushBot(N, \texttt{bot\_min})$
        \State $N.\texttt{bot\_lazy} \gets $ \texttt{false}
        \EndIf
        \State \Return $N$
        \EndFunction
        
\\
        
        \Function{PushTop}{$N, y$} \Comment{Returns new version of $N$ with updated children}
        \State $L \gets $ copy of left child of $N$
        \State $R \gets $ copy of right child of $N$
        \State $L.\texttt{top\_min} \gets y$, $L.\texttt{top\_max} \gets y$, $L.\texttt{top\_lazy} \gets \texttt{true}$
        \State $R.\texttt{top\_min} \gets y$, $R.\texttt{top\_max} \gets y$, $R.\texttt{top\_lazy} \gets \texttt{true}$
        \State Recalculate \texttt{top\_hash} and \texttt{hash} for $R$ and $L$.
        \State $N \gets $ copy of $N$
        \State $N.\texttt{left} \gets L$, $N.\texttt{right} \gets R$
        \State $N.\texttt{top\_lazy} \gets \texttt{false}$
        \State \Return $N$
        \EndFunction

\\

        \Function{PushBot}{$N, y$} \Comment{Same as above, omitted for clarity}
        \EndFunction
    \end{algorithmic}
\end{algorithm}

\subparagraph*{\Mark{} operation. }
Let $v_l,\ldots,v_r$ be the points between left and right side of the query square, i.e. the points potentially affected by this update. If there are no such points, then we ignore this update. Assume without loss of generality that the query square is a \textit{bottom} square, so we need to update the bottom boundary.

As usual, we proceed recursively. Let $\UpdateBot(N,[l,r],y)$ be a procedure which updates the boundary of the node $N$ with the update square. Let the node $N$ be responsible for the points $v_a,\ldots,v_b$.

If $[a,b]\cap [l,r] = \varnothing$, then this update does not affect the boundary in that node, therefore we can safely return.

If $[a,b]\subseteq [l,r]$ and the node is $y$-disjoint (we can check that using the stored values), then we can directly update the values stored in the node $N$. Let us introduce the names $B, T,Y$ for the  bottom boundary, top boundary and horizontal line at $y$ level, respectively. We need these names for a few cases that we will consider, and these cases depend on relative locations of $B$, $T$ and $Y$:

\begin{enumerate}
    \item If $Y < B$ (i.e. $y < \texttt{bot\_min}$), then we are marking a region that is below the bottom boundary, so we do not need to update any information and we can safely return.
    \item If $B < Y < T$ (i.e. $\texttt{bot\_max} < y < \texttt{top\_min}$) then the new bottom boundary is going to be $Y$. We can set $\texttt{bot\_min} = \texttt{bot\_max} = y$ and calculate $\texttt{bot\_hash}$ using \texttt{pref\_hash} table (marked points are prefix of \texttt{points} array). We calculate \texttt{hash} as $\texttt{bot\_hash} \oplus \texttt{top\_hash}$ (marked regions are disjoint in this case). We mark the node $N$ as bottom-lazy, because nodes in its subtree might not have correct values. Node $N$ is now bottom-simple and split, so invariants are maintained.
    \item If $B,T < Y$ (i.e. $\texttt{bot\_max}, \texttt{top\_max} < y$) then the new bottom boundary is above the top boundary, so every point is marked and we can set \texttt{hash} to the hash of all points.
\end{enumerate}

We can differentiate between the cases using the information from the node as shown above. After these calculations all invariants are maintained and we can safely return.

If none of the above conditions holds, we need to use recursion on children. However, children might be outdated if the node $N$ is lazy. In that case, we first execute $\Push(N)$ to fix values in the children and only then we proceed recursively.
Later, we merge information about the boundaries and marked points. We use the following formulas:

\begin{itemize}
    \item $\texttt{top\_min} = \min(L.\texttt{top\_min}, R.\texttt{top\_min})$,
    \item $\texttt{top\_max} = \max(L.\texttt{top\_max}, R.\texttt{top\_max})$,
    \item $\texttt{bot\_min} = \min(L.\texttt{bot\_min}, R.\texttt{bot\_min})$,
    \item $\texttt{bot\_max} = \max(L.\texttt{bot\_max}, R.\texttt{bot\_max})$,
    \item $\texttt{top\_hash} = L.\texttt{top\_hash} \oplus R.\texttt{top\_hash}$,
    \item $\texttt{bot\_hash} = L.\texttt{bot\_hash} \oplus R.\texttt{bot\_hash}$,
    \item $\texttt{hash} = L.\texttt{hash} \oplus R.\texttt{hash}$,
    \item $\texttt{top\_lazy}$ and $\texttt{bot\_lazy}$ are already \texttt{false}, as we used $\Push$ on them if they were not.
\end{itemize}

\begin{algorithm}
    \caption{$\Mark$ operation}
    \begin{algorithmic}
        \Function{Mark}{$\widetilde{S}, (x,y)$}
        \If {square centered at $(x,y)$ cannot affect any point}
        \State \Return $\widetilde{S}$
        \EndIf
        \State Calculate $l,r$ such that points $v_l,\ldots,v_r$ are between left and right side of query square.
        \If{$y \leq 0.5$} \Comment{Query square intersects bottom side of the stripe.}
        \State \Return $\UpdateBot(\widetilde{S},[l,r],y+0.5)$
        \Else \Comment{Query square intersects top side of the stripe.}
        \State \Return $\UpdateTop(\widetilde{S},[l,r],y-0.5)$
        \EndIf
        \EndFunction
        \Function{UpdateBot}{$N, [l,r], y$}
        \State Let $[a,b]$ be the interval associated with node $N$.
        \If{$[a,b]\cap [l,r] = \varnothing$}
            \State \Return $N$
            \EndIf
            \State $N \gets $ copy of $N$
            \If{$[a,b]\subseteq [l,r]$ and $N$ is $y$-disjoint}
            \State Update values of $N$.
            \State \Return $N$.
            \EndIf
            \If {$N$ is lazy}
            \State $N \gets \Push(N)$
            \EndIf
            \State $N.left \gets \UpdateBot(N.\texttt{left}, l, r, y)$
            \State $N.right \gets \UpdateBot(N.\texttt{right}, l, r, y)$
            \State Merge information from children to $N$.
            \State \Return $N$
        \EndFunction
    \end{algorithmic}
\end{algorithm}

\subparagraph*{\ListDifferences{} operation.}
Given representatives of two sets (i.e. roots of some two versions of the tree), we want to list the symmetric difference between sets of marked points. To compare content of the sets we use hashes stored in the nodes. We proceed by concurrently descending both trees. If the hashes of the sets are equal then the symmetric difference is empty (with high probability). Now suppose that the hashes are different. If the node is a leaf with a single-element interval $v_i$ then we add $v_i$ to the result. Otherwise, we split recursively and return the union of sets returned from recursive calls.

To make our hashes work, we set $t$ (number of bits in hashes) to be $c\log_2{n}$. The well-known arguments shows that some two different sets may have the same hash with probability no larger than $2^{-c\log_2{n}} = n^{-c}$, i.e. we can correctly identify sets as different with high probability.
Therefore, if hashes are the same, then w.h.p. sets are the same.


\begin{algorithm}
    \caption{$\ListDifferences$ algorithm}
    \begin{algorithmic}
        \Function{ListDifferences}{$N_1,N_2$}
        \If {$N_1$ is lazy}
        \State $N_1 \gets \Push(N_1)$
        \EndIf
        \If {$N_2$ is lazy}
        \State $N_2 \gets \Push(N_2)$
        \EndIf
        \If{$N_1.hash = N_2.hash$} \Comment{Sets are the same, symmetric difference is empty.}
        \State \Return $\varnothing$
        \EndIf
        \If{$N_1,N_2$ are leaves} \Comment{Leaves responsible for point $P$.}
        \State \Return $\{P\}$ \Comment{One set is empty, the other is $\{P\}$.}
        \EndIf
        \State $D_1 \gets \ListDifferences(N_1.\texttt{left}, N_2.\texttt{left})$
        \State $D_2 \gets \ListDifferences(N_1.\texttt{right}, N_2.\texttt{right})$
        \State \Return $D_1 \cup D_2$
        \EndFunction
    \end{algorithmic}
\end{algorithm}

\subsubsection*{Time complexity analysis}

\subparagraph*{Initialization.}
The height of the tree is $\Oh{\log{n}}$ and there are $\Oh{n}$ nodes (as in every binary tree with $n$ leaves), so the initialization of \texttt{points} and \texttt{pref\_hash} arrays (whose size is $\Oh{n\log{n}}$ in total) can be done in $\Ot{n}$ time.
We set all lazy flags to false, all \texttt{top\_min/max} to 1, all \texttt{bot\_min/max} to 0, all \texttt{hash}, \texttt{bot\_hash} and \texttt{top\_hash} to 0 (hash of empty set). This takes $\Ot{n}$ time in total.

\subparagraph*{\Push{} operation.}
This operation uses constant number of basic operations plus manipulations of hashes, therefore it is $\Oh{t} = \Ot{1}$. (Recall that $t$ is the hash length, so it will come up in every operation. However, as $t = \Oh{\log n}$, it does not change the complexity meaningfully).

\subparagraph*{\Mark{} operation.}
The affected interval $[l;r]$ can be found in $\Ot{1}$ using binary search. It remains to establish the time complexity of the calls to $\UpdateBot$ or $\UpdateTop$.
We analyze complexity of $\UpdateBot$ operation as $\UpdateTop$ is analogous.
Let $Y$ denote the top side of the query square and $y$ denote its $y$-coordinate.

Excluding recursive calls, we make $\Oh{t}$ operations inside $\UpdateBot$ function, so we need only to calculate how many recursion calls we perform. There are two cases when we use recursion: when $[a,b] \cap [l,r]$ is neither $[a,b]$ nor $\varnothing$ and when $[a,b] \subseteq [l,r]$ and the node is not $y$-disjoint. We will calculate how many times we use recursion in the first and in the second case.

\begin{lemma}
    The $\UpdateBot$ procedure visits $\Oh{\log{n}}$ nodes with $[l,r] \cap [a,b]$ being neither $[a,b]$ nor $\varnothing$.
\end{lemma}
\begin{proof}
    Notice that the nodes at the same level are always responsible for disjoint intervals of points. The conditions imply that if we use recursion at node $N$, then the interval which $N$ is responsible for must contain one of $l,r$. Therefore there are at most $2$ nodes at the same level for which we use recursion. Conclusion follows. 
\end{proof}

Now we analyze structure of boundaries to bound the number of times they can intersect with a horizontal segment (non-$y$-disjoint node implies intersection of horizontal line with boundary).

\begin{lemma}
    For any $r\in\mathbb{R}$, the bottom boundary is weakly bitonic when restricted to the interval $[r,r+1]$, i.e. there are no 3 points $p_1,p_2,p_3$ on the boundary\,\footnote{Note that the points in the lemma are any boundary points, not necessarily points stored in the structure.} such that $r \leq p_1^{(x)} < p_2^{(x)} < p_3^{(x)} \leq r+1$ and $p_1^{(y)} < p_2^{(y)} > p_3^{(y)}$.

\end{lemma}
\begin{proof}
    Suppose that this is not the case. Then the distance between $p_1^{(x)}$ and $p_3^{(x)}$ is at most 1. Now, the square added by $\Mark$ covering $p_2$ had width 1, therefore it must have covered also $p_1$ or $p_3$. But then, as $p_2$ is strictly higher than both of them, $p_1$ or $p_3$ cannot be on the boundary -- a contradiction. 
\end{proof}

Analogously we can prove a similar claim for the top boundary.
\begin{observation}
\label{obs:intersections}
    In a single \Mark call, any boundary can intersect the top segment $Y$ of the query square in at most 2 points.
\end{observation}
\begin{proof}
    Assume that it is the bottom boundary and note that $Y$ has length $1$, so we can apply the previous lemma. Suppose by the contrary that 3 leftmost intersections are at points $p_1,p_2,p_3$ where $p_1^{(x)} < p_2^{(x)} < p_3^{(x)}$. Now, if any boundary point $p$ at interval $(p_1^{(x)};p_2^{(x)})$ is above $Y$ segment, then $p_1^{(x)} < p^{(x)} < p_2^{(x)}$ and $p_1^{(y)} < p^{(y)} > p_2^{(y)}$, contradiction. We proceed similarly if any point at interval $(p_2^{(x)};p_3^{(x)})$ is above $Y$ segment. Therefore any boundary points $p\in (p_1^{(x)};p_2^{(x)}), q\in (p_2^{(x)};p_3^{(x)})$ are below $Y$ segment. But then $p^{(x)} < p_2^{(x)} < q^{(x)}$ and $p^{(y)} < p_2^{(y)} > q^{(y)}$ -- a contradiction.
\end{proof}

For any $\Mark$ operation we have at most $4$ intersections of horizontal segment $Y$ with both boundaries. We conclude that at every level of the tree there are at most 4 non-$y$-disjoint nodes. In total, there are $\Oh{\log{n}}$ non-$y$-disjoint nodes. 

Non-$y$-disjoint nodes for which $[a,b]\subseteq [l,r]$ (i.e. nodes where we use recursion) are subset of non-$y$-disjoint nodes. But non-$y$-disjoint nodes are exactly those which contain an intersection of $Y$ with some boundary, therefore there are at most $\Oh{\log{n}}$ such nodes.

In total, we use recursion $\Oh{\log{n}}$ times, so time complexity of $\Mark$ operation is $\Oh{t\log{n}} = \Ot{1}$.

\subparagraph*{\ListDifferences{}.} To analyze the complexity of a single call of this operation, let us put a token on all leaves which contribute to resulting symmetric difference. There are exactly $|D|$ such leaves, therefore there are $\Oh{|D|\log{n}}$ nodes that are ancestors of at least one on them. Observe that we will use recursion only when invoked on such nodes, otherwise we will simply return. Therefore, the time complexity of $\ListDifferences$ operation is $\Oh{t|D|\log{n}} = \Ot{|D|}$.

\subsection{Full data structure for unit squares}
\label{sec:squares-geo-structure}

Here we prove the Lemma \ref{lemma:structure-geo}a), i.e. for $\mathcal{F}$ being an axis-aligned unit square, therefore proving Theorem \ref{thm:polygons}a) for intersection graph of unit squares.

\begin{proof}[Proof of Lemma \ref{lemma:structure-geo}a)]

Let us divide the plane into stripes of height $1$. Let us consider only the stripes containing at least one point from $V$, denoting them $Q_1, Q_2,\cdots, Q_r$. For every such stripe, we construct its own Single Stripe Data Structure from Subsection \ref{sec:stripes-structure}. This step takes a total of $\Ot{n}$ time, as initialization of every stripe takes time proportional to number of points inside this stripe and there are $n$ points in total.

Consider an operation $\Mark()$ with some square $H$. The square $H$ intersects at most $2$ stripes from our collection. We can update these stripes separately to mark appropriate points inside them in the same complexity as $\Mark()$ from SSDS, i.e. $\Ot{1}$.

There is, however, a problem with the $\ListDifferences$ operation. Given two collections of stripes from different moments of time we would like to list the symmetric difference between the sets of marked points. However, we cannot iterate over all stripes and run $\ListDifferences$ on every one of them, as it would take $\Ot{r + |D|}$ time, where $r$ is number of stripes. This can be $\Omega(n)$ in worst case, which is way more than we need. Hence, we require some way to identify the stripes where there is at least one point belonging to the resulting symmetric difference.

To achieve that, we use an auxiliary structure virtually identical to Simple Set Retrieval from Section \ref{sec:persistent-segment-tree}. We associate the leaves of the segment tree (i.e. elements of the underlying set) with the stripes $Q_1, \ldots, Q_r$. Each leaf stores a pointer to the representative of a stripe and the hash of its marked points. This way we can efficiently compare two sets of stripes and identify leaves with different hashes, which are the stripes with non-empty symmetric difference between two versions. This allows us to determine stripes we need in time $\Ot{|D|}$. We also need to update the SSR structure appropriately after each $\Mark$ operation.

\end{proof}

\subsection{Data structure for convex polygons}
\label{sec:full-geo-structure}

Here we prove the Lemma \ref{lemma:structure-geo}b) for $\mathcal{F}$ being an $s$-sided convex polygon with center of symmetry.
By Lemma \ref{lemma:symmetric} the Theorem \ref{thm:polygons}b) follows.

\begin{proof}[Proof of Lemma \ref{lemma:structure-geo}b]

The symmetry implies that the sides of $\mathcal{F}$ can be split into parallel pairs. Assume that longest side of $\mathcal{F}$ has length 1 and rotate the whole plane so that the chosen side is parallel to the $OY$ axis. Note that the height of $\mathcal{F}$ is now at most $s$.

We want to further transform plane before proceeding. Let $l,l'$ be pair of sides parallel to $OY$ axis. Consider affine transformation which sends $l,l'$ to sides of an axis-aligned unit square. Such an affine transformation preserves length of segments parallel to $OY$, therefore image of $\mathcal{F}$ under this transformation also has height at most $s$. Additionally, affine transformations preserve central symmetry. Now we can assume that $\mathcal{F}$ has center of symmetry, height at most $s$, two sides $l,l'$ of length 1 parallel to $OY$ axis and contains an unit square with sides $l,l'$. 


We now partition the polygon $\mathcal{F}$ into a set of $\Oh{s}$ trapezoids using vertical lines.
More precisely, we apply a vertical cut going through vertex $v$ for each vertex $v$.
Since the transformed shape contains an axis-aligned unit square, non-vertical sides of each trapezoid cannot intersect the same stripe. 

We use the same division of plane into horizontal stripes as in the unit-square case. The boundaries of stripes partition each trapezoid into few ($\Oh{s}$ for every trapezoid) parts. In total, there are at most $\Oh{s^2}$ parts of the polygon and let us assume that \Mark operation will deal with every part of the polygon separately. But each part has now a similar property as before -- it can be assigned either to bottom or top boundary of some stripe. Each part is either a rectangle with both sides at the boundaries of the stripe or a trapezoid with single side at boundary of stripe and two sides perpendicular to that boundary (see Figure \ref{fig:trapezoids}).

\begin{figure}
    \begin{center}
        \includegraphics[scale=1]{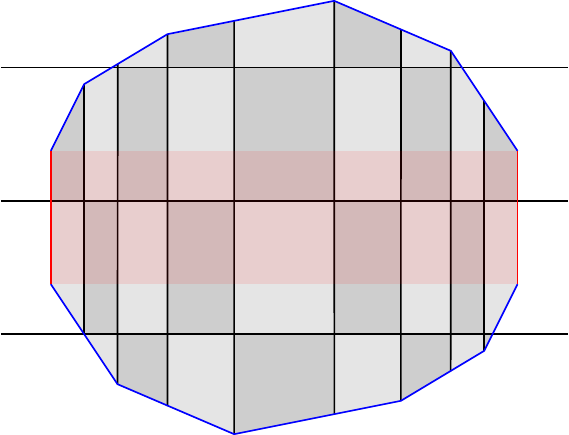}
    \end{center}
    \caption{Partition of a polygon into trapezoids. For clarity the picture was stretched along $OX$ axis. (In particular, the red rectangle is in fact a unit square.)}
    \label{fig:trapezoids}
\end{figure}

\begin{figure}
    \begin{center}
        \includegraphics[scale=1]{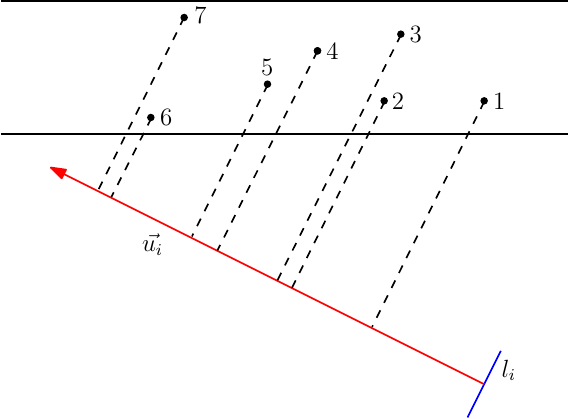}
    \end{center}
    \caption{Order of points with respect to direction $\vec{u_i}$ determined by side $l_i$.}
    \label{fig:dir-sorting}
\end{figure}

Let $l_1 ,l_2, \ldots, l_s$ be sides of $\mathcal{F}$. Let $\vec{u_1}, \vec{u_2}, \ldots,\vec{u_s}$ be vectors perpendicular to the sides $l_1, \ldots, l_s$ respectively, pointing outside of $\mathcal{F}$.
We modify nodes of the Single Stripe Data Structure from Section \ref{sec:stripes-structure} to hold the following information:

\begin{itemize}
\item For every $1\leq i\leq s$, each node stores list of its points sorted along $\vec{u_i}$ (Figure \ref{fig:dir-sorting}). For every such list we also keep appropriate list with prefix hashes;
\item For every $1\leq i\leq s$, each node maintains the extremes of its bottom and top boundary in direction $\vec{u_i}$. 
\end{itemize}

Additionally, we need to slightly change the definition of a simple node: a node is called bottom-simple (top-simple) if its bottom boundary (top-boundary) is a segment (not necessarily horizontal, only parallel to some side of $\mathcal{F}$).

The initialization and $\ListDifferences$ operation can be easily adapted to the new setting.
The non-trivial part is the $\Mark$ operation, which is now given a trapezoid.
We handle them similarly: we descend recursively until we find nodes that can be directly updated.
We just need a way to decide whether the side of trapezoid is disjoint from both boundaries.
This can be tested by comparing the extreme points of the both boundaries for direction $\vec{u_i}$ perpendicular to the trapezoid's side.

We need to show that time complexity of all operations is still within expected bounds. Operation $\Push$ is affected only by the fact that it needs to update $\Oh{s}$ values now, but complexity stays $\Ot{1}$ for constant $s$. It follows that $\ListDifferences$ also stays at $\Ot{|D|}$ time complexity.

Now we analyze complexity of $\Mark$. Inspection of time complexity proof from Section \ref{sec:stripes-structure} shows that we only need to check if number of intersections of boundary with any trapezoid is still constant (maybe dependent on $s$). Indeed, this is the case.

Let $s=2g$ and denote consecutive vertices of $\mathcal{F}$ as $w_1,w_2,\ldots,w_{2g}$. For $1 \leq i \leq g$ let $l_i$ be side between vertices $w_i$ and $w_{i+1}$ translated so that $w_i$ coincides with $(0,0)$. Now $(0,0) \in l_i$ for each $i$. Simple induction shows that \[\mathcal{F} \equiv l_1\oplus l_2 \oplus \ldots \oplus l_g.\]
For every $1 < j \leq g$ we get $l_1 \oplus l_j \subseteq l_1\oplus l_2 \oplus \ldots \oplus l_g$. We can now proceed to prove that each trapezoid intersects in at most $2$ points with any boundary.

Suppose that non-vertical side of trapezoid is congruent to $l_j$ for some $j$. Without loss of generality assume that this side intersects bottom boundary at segment originating from side $l_1$ of some already marked copy of $\mathcal{F}$ and $\mathcal{F}$ is at right side of line containing $l_1$. We will show that right end of side of the trapezoid is contained in marked area below bottom boundary. From previous paragraph parallelogram spanned by $l_1$ and $l_j$ placed alongside of $l_1$ of $\mathcal{F}$ is fully contained inside $\mathcal{F}$. Clearly this parallelogram contains right end of side of trapezoid. By convexity, every point on side of trapezoid from intersection point to its right end is already marked. This shows that side can intersect with any boundary at most 2 times.

This implies that complexity of $\Mark$ operation for $s$-sided polygons remains $\Ot{1}$. Conclusion follows.

\end{proof}

\end{document}